\documentclass[a4paper,UKenglish,cleveref, autoref]{lipics-v2019}

\usepackage{amsmath}
\usepackage{amsthm}
\usepackage{amssymb}
\usepackage{algorithm}
\usepackage{color}
\usepackage{wrapfig,epsfig}
\usepackage{epstopdf}
\usepackage{url}
\usepackage{graphicx}
\usepackage{color}
\usepackage{epstopdf}
\usepackage{algpseudocode}
\usepackage{scrextend}
\usepackage[T1]{fontenc}
\usepackage{bbm}
\usepackage{comment}

\newtheorem{conjecture}{Conjecture}

\title{Log Diameter Rounds Algorithms for $2$-Vertex and $2$-Edge Connectivity} 

\titlerunning{Log Diameter Rounds Algorithms for $2$-Vertex and $2$-Edge Connectivity}

\author{Alexandr Andoni}{Columbia University}{andoni@cs.columbia.edu}{}{Research partly supported by NSF Grants (CCF-1617955 and CCF-1740833), Simons Foundation (\#491119) and Google Research Award.}

\author{Clifford Stein}{Columbia University}{cliff@cs.columbia.edu}{}{Research partly supported by NSF Grants  CCF-1714818 and CCF-1822809.}

\author{Peilin Zhong}{Columbia University}{pz2225@columbia.edu}{}{Research partly supported by NSF Grants (CCF-1703925, CCF-1421161, CCF-1714818, CCF-1617955 and CCF-1740833), Simons Foundation (\#491119) and Google Research Award.}

\authorrunning{A. Andoni, C. Stein and P. Zhong}

\Copyright{Alexandr Andoni, Clifford Stein and Peilin Zhong}

\ccsdesc[300]{Theory of computation~MapReduce algorithms}
\ccsdesc[300]{Mathematics of computing~Paths and connectivity problems}

\keywords{parallel algorithms, biconnectivity, $2$-edge connectivity, the MPC model}

\category{}

\relatedversion{}

\supplement{}

\acknowledgements{}

\nolinenumbers 

\hideLIPIcs  

\EventEditors{Christel Baier, Ioannis Chatzigiannakis, Paola Flocchini, and Stefano Leonardi}
\EventNoEds{4}
\EventLongTitle{46th International Colloquium on Automata, Languages, and Programming (ICALP 2019)}
\EventShortTitle{ICALP 2019}
\EventAcronym{ICALP}
\EventYear{2019}
\EventDate{July 9--12, 2019}
\EventLocation{Patras, Greece}
\EventLogo{eatcs}
\SeriesVolume{132}
\ArticleNo{9}

\newcommand{\wh}{\widehat}
\newcommand{\wt}{\widetilde}

\renewcommand{\varepsilon}{\epsilon}
\renewcommand{\tilde}{\wt}
\renewcommand{\hat}{\wh}

\DeclareMathOperator{\p}{par}

\DeclareMathOperator{\rank}{rank}

\DeclareMathOperator{\dist}{dist}
\DeclareMathOperator{\cyclen}{cyclen}

\DeclareMathOperator{\dep}{dep}

\DeclareMathOperator{\lef}{left}
\DeclareMathOperator{\rig}{right}
\DeclareMathOperator{\diam}{diam}
\DeclareMathOperator{\bidiam}{bi-diam}

\DeclareMathOperator{\col}{col}
\DeclareMathOperator{\lca}{lca}
\DeclareMathOperator{\bac}{lev}
\DeclareMathOperator{\rmq}{rmq}
\DeclareMathOperator{\child}{child}
\DeclareMathOperator{\leaves}{leaves}

\begin{document}

\maketitle

\begin{abstract}
Many modern parallel systems, such as MapReduce, Hadoop and Spark, can be modeled well by the MPC model.
The MPC model captures well coarse-grained computation on large data --- data is distributed to processors,
each of which has a sublinear (in the input data) amount of memory and we alternate between rounds of
computation and rounds of communication, where each machine can communicate an amount of data as large
as the size of its memory.  This model is stronger than the classical PRAM model, and it is an intriguing
question to design algorithms whose running time is smaller than in the PRAM model.

In this paper, we study two fundamental problems, $2$-edge connectivity and $2$-vertex connectivity (biconnectivity).
PRAM algorithms which run in $O(\log n)$ time have been known for many years.  We give algorithms using roughly log diameter
rounds in the MPC model.     
Our main results are, for an $n$-vertex, $m$-edge graph of diameter $D$ and bi-diameter $D'$, 1) a 
$O(\log D\log\log_{m/n} n)$ parallel time $2$-edge connectivity
algorithm, 2) a $O(\log D\log^2\log_{m/n}n+\log
D'\log\log_{m/n}n)$ parallel time biconnectivity algorithm, 
where the  bi-diameter $D'$ is the largest cycle length over all the vertex pairs in the same biconnected component.
Our results are fully scalable, meaning that 
the memory per processor can be $O(n^{\delta})$ for arbitrary constant $\delta>0$, and the total memory used is linear in the problem size. 
Our
$2$-edge connectivity algorithm achieves the same parallel time as the
connectivity algorithm of~\cite{asswz18}.  We also show an $\Omega(\log
D')$ conditional lower bound for the biconnectivity problem.

\end{abstract}

\section{Introduction}
The success of modern parallel and distributed systems such as
MapReduce~\cite{dg04,dg08}, Spark~\cite{zcfss10}, Hadoop~\cite{w12},
Dryad~\cite{ibybf07}, together with the need to solve problems on
massive data, is driving the development of new algorithms which are
more efficient and scalable in these large-scale systems.  An
important theoretical problem is to develop models which are good
abstractions of these computational frameworks.  The \emph{Massively
  Parallel Computation (MPC)}
model~\cite{ksv10,gsz11,bks13,anoy14,bks15,clmmos18,asswz18} 
captures the capabilities of these computational systems while keeping the
description of the model itself simple.  In the MPC model, there are
machines (processors), each with $\Theta(N^{\delta})$ local
memory, where $N$ denotes the size of the input and $\delta\in(0,1)$.
The computation proceeds in rounds, where each machine can perform
unlimited local computation in a round and exchange $O(N^{\delta})$
data at the end of the round.  The parallel time of an algorithm is
measured by the total number of computation-communication rounds.  The
MPC model is a variant of the Bulk Synchronous Parallel (BSP)
model~\cite{v90}.  It is also a more powerful model than the PRAM
since any PRAM algorithm can be simulated in the MPC model
~\cite{ksv10,gsz11} while some problem can be solved in a
faster parallel time in the MPC model.  For example, computing the XOR of
$N$ bits takes $O(1/\delta)$ parallel time in the MPC model but needs
near-logarithmic parallel time on the most powerful CRCW
PRAM~\cite{bh89}.

A natural question to ask is: which problems can be solved in 
faster parallel time in the MPC model than on a 
PRAM?  This question has been studied by a line of
recent
papers~\cite{ksv10,eim11,lmsv11,anoy14,ag18,ak17,ims17,clmmos18,asw18,bfu18,bdhk18,o18,fgu18}.
Most of these results studied the graph problems, which are the usual
benchmarks of parallel/distributed models.  Many graph problems such
as graph connectivity~\cite{sv80,r85,lt18}, graph
biconnectivity~\cite{tv84,tv85}, maximal matching~\cite{kuw86},
minimum spanning tree~\cite{kprs97} and maximal independent
set~\cite{l86,abi86} can be solved in the standard logarithmic time in
the PRAM model, but these  problems have been shown to have a better parallel
time in the MPC model.  

In addition, we hope to develop \emph{fully scalable} algorithms for
the graph problems, i.e., the algorithm should work for any constant
$\delta>0$.  The previous literatures show that a graph problem in the
MPC model with large local memory size may be much easier than the
same problem in the MPC model but with a smaller local memory size.
In particular, when the local memory size per machine is close to the
number of vertices $n$, many graph problems have efficient algorithms.
For example, if the local memory size per machine is $n/\log^{O(1)}n$,
the connectivity problem~\cite{asw18} and the approximate matching
problem~\cite{abbms19} can be solved in $O(\log\log n)$ parallel time.
If the local memory size per machine is $\Omega(n)$, then the MPC
model meets the congested clique model~\cite{bdh18}.  In this setting,
the connectivity problem and the minimum spanning tree problem can be
solved in $O(1)$ parallel time~\cite{jn18}.  If the local memory size
per machine is $n^{1+\Omega(1)}$, many graph problems such as maximal
matching, approximate weighted matchings, approximate vertex and edge
covers, minimum cuts, and the biconnectivity problem can be solved in $O(1)$ parallel
time~\cite{lmsv11,aflp12}.  The landscape of graph algorithms in the MPC
model with small local memory is more nuanced and challenging for algorithm
designers.  If the local
memory size per machine is $n^{1-\Omega(1)}$, then the best
connectivity algorithm takes parallel time $O(\log D\log\log n)$ where
$D$ is the diameter of the graph~\cite{asswz18}, and the best
approximate maximum matching algorithm takes parallel time
$\tilde{O}(\sqrt{\log n})$~\cite{o18}.

Therefore, the main open question is: which kind of the graph problems can have faster fully scalable MPC algorithms than the standard logarithmic PRAM algorithms?

Two fundamental graph problems in graph theory are $2$-edge
connectivity and $2$-vertex connectivity (biconnectivity).  In this
work, we studied these two problems in the MPC model.  Consider an
$n$-vertex, $m$-edge undirected graph $G$.  A bridge of $G$ is an edge
whose removal increases the number of connected components of $G$.  In
the $2$-edge connectivity problem, the goal is to find all the bridges
of $G$.  For any two different edges $e,e'$ of $G$, $e,e'$ are in the
same biconnected component (block) of $G$ if and only if there is a
simple cycle which contains both $e,e'$.  If we define a relation $R$ such
that $e R e'$ if and only if $e=e'$ or $e,e'$ are contained by a
simple cycle, then $R$ is an equivalence relation~\cite{d18}.  Thus, a
biconnected component is an induced graph of an equivalence class of
$R$.  In the biconnectivity problem, the goal is to output all the
biconnected components of $G$.  We proposed faster, fully scalable
algorithms for the both $2$-edge connectivity problem and the
biconnectivity problem by parameterizing the running time as a
function of the \emph{diameter} and the \emph{bi-diameter} of the
graph. The diameter $D$ of $G$ is the largest diameter of its
connected components.  The definition of bi-diameter is a natural
generalization of the definition of diameter.  If vertices $u,v$ are
in the same biconnected component, then the cycle length of $(u,v)$ is
defined as the minimum length of a simple cycle which contains both
$u$ and $v$. The bi-diameter $D'$ of $G$ is the largest cycle length
over all the vertex pairs $(u,v)$ where both $u$ and $v$ are in the
same biconnected component.  Our main results are 1) a fully scalable
$O(\log D\log\log_{m/n} n)$ parallel time $2$-edge connectivity
algorithm, 2) a fully scalable $O(\log D\log^2\log_{m/n}n+\log
D'\log\log_{m/n}n)$ parallel time biconnectivity algorithm.  Our
$2$-edge connectivity algorithm achieves the same parallel time as the
connectivity algorithm of~\cite{asswz18}.  We also show an $\Omega(\log
D')$ conditional lower bound for the biconnectivity problem.

\subsection{The Model}
Our model of computation is 
the Massively Parallel Computation (MPC) model~\cite{ksv10,gsz11,bks13}. 

Consider two non-negative parameters $\gamma\geq 0,\delta>0$.
In the $(\gamma,\delta)$-MPC model~\cite{asswz18}, there are $p$ machines (processors) each with local memory size $s$, where $p\cdot s=\Theta(N^{1+\gamma}),s=\Theta(N^{\delta})$ and $N$ denotes the size of the input data. 
Thus, the space per machine is sublinear in $N$, and the total space is only an $O(N^{\gamma})$ factor more than the input size.
In particular, if $\gamma=0$, the total space available in the system is linear in the input size $N$.
The space size is measured by words each containing $\Theta(\log(s\cdot p))$ bits.
Before the computation starts, the input data is distributed on $\Theta(N/s)$ input machines. 
The computation proceeds in rounds. 
In each round, each machine can perform local computation on its local data, and send messages to other machines at the end of the round.
In a round, the total size of messages sent/received by a machine should be bounded by its local memory size $s=\Theta(N^{\delta})$.
For example, a machine can send $s$ size $1$ messages to $s$ machines or send a size $s$ message to $1$ machine in a single round. 
However, it cannot broadcast a size $s$ message to every machine.
In the next round, each machine only holds the received messages in its local memory.
At the end of the computation, the output data is distributed on the output machines.
An algorithm in this model is called a $(\gamma,\delta)$-MPC algorithm. 
The parallel time of an algorithm is the total number of rounds needed to finish its computation.
In this paper, we consider $\delta$ an arbitrary constant in $(0,1)$.

\subsection{Our Results}
Our main results are efficient MPC algorithms for $2$-edge
connectivity and biconnectivity problems.
In our algorithms, one important subroutine is 
computing the Depth-First-Search (DFS) sequence~\cite{asswz18} which is
a variant of the Euler tour representation proposed
by~\cite{tv84,tv85} in 1984.  We show how to efficiently compute the
DFS sequence in the MPC model with linear total space.  Conditioned on
the hardness of the connectivity problem in the MPC model, we prove a
hardness result on the biconnectivity problem.

For $2$-edge connectivity and biconnectivity, the input is an undirected graph $G=(V,E)$ with $n=|V|$ vertices and $m=|E|$ edges. $N=n+m$ denotes the size of the representation of $G$, $D$ denotes the diameter of $G$, and $D'$ denotes the bi-diameter of $G$.
We state our results in the following. 

\noindent\textbf{Biconnectivity.} In the biconnectivity problem, we want to find all the biconnected components (blocks) of the input graph $G$.
Since the biconnected components of $G$ define a partition on $E$, we just need to color each edge, i.e., at the end of the computation, $\forall e\in E$, there is a unique tuple $(x,c)$ with $x=e$ stored on an output machine, where $c$ is called the color of $e$, such that the edges $e_1,e_2$ are in the same biconnected components if and only if they have the same color.

\begin{theorem}[Biconnectivity in MPC]\label{thm:final_biconn}
	For any $\gamma\in[0,2]$ and any constant $\delta\in(0,1)$, there is a randomized $(\gamma,\delta)$-MPC algorithm which outputs all the biconnected components of the graph $G$ in $O\left(\log D\cdot\log^2\frac{\log n}{\log(N^{1+\gamma}/n)}+\log D'\cdot\log\frac{\log n}{\log(N^{1+\gamma}/n)}\right)$ parallel time. The success probability is at least $0.95$. If the algorithm fails, then it returns FAIL.
\end{theorem}
The worst case is when the input graph is sparse and the total space available is linear in the input size, i.e., $N=n+m=O(n)$ and $\gamma=0$. In this case, the parallel running time of our algorithm is $O(\log D\cdot \log^2\log n + \log D'\cdot \log\log n)$. If the graph is slightly denser ($m=n^{1+c}$ for some constant $c>0$), or the total space is slightly larger ($\gamma>0$ is a constant), then we obtain $O(\log D+\log D')$ time.

A cut vertex (articulation point) in the graph $G$ is a vertex whose removal increases the number of connected components of $G$.
Since a vertex $v$ is a cut vertex if and only if there are two edges $e_1,e_2$ which share the endpoint $v$ and $e_1,e_2$ are not in the same biconnected component, our algorithm can also find all the cut vertices of $G$.

\noindent\textbf{$2$-Edge connectivity.} In the $2$-edge connectivity
problem, we want to output all the bridges of the input graph $G$.  
Since an edge is a bridge if and only if
each of its endpoints is either a cut vertex or a vertex with degree
$1$, the $2$-edge connectivity problem should be easier than the
biconnectivity problem.  We show how to solve $2$-edge connectivity in
the same parallel time as the algorithm proposed by~\cite{asswz18} for
solving connectivity.

\begin{theorem}[$2$-Edge connectivity in MPC]\label{thm:two_edge_connectivity}
	For any $\gamma\in[0,2]$ and any constant $\delta\in(0,1)$, there is a randomized $(\gamma,\delta)$-MPC algorithm which outputs all the bridges of the graph $G$ in $O\left(\log D\cdot\log\frac{\log n}{\log(N^{1+\gamma}/n)}\right)$ parallel time. The success probability is at least $0.97$. If the algorithm fails, then it returns FAIL.
\end{theorem}

\noindent\textbf{DFS sequence.} 
A rooted tree with a vertex set $V$ can be represented by $n=|V|$ pairs $(v_1,\p(v_1)),(v_2,\p(v_2)),\cdots,(v_n,\p(v_n))$ where $\p:V\rightarrow V$ is a set of parent pointers, i.e.,   
for a non-root vertex $v$, $\p(v)$ denotes the parent of $v$, and for the root vertex $v$, $\p(v)=v$.
We show an algorithm which can compute the DFS sequence (Definition~\ref{def:dfs_sequence}) of the rooted tree in the MPC model with linear total space.

 \begin{theorem}[DFS sequence of a tree in MPC]\label{thm:DFS_sequence}
	Given a rooted tree represented by a set of parent pointers $\p:V\rightarrow V$, there is a randomized $(0,\delta)$-MPC algorithm which outputs the DFS sequence in $O(\log D)$ parallel time, where $\delta\in(0,1)$ is an arbitrary constant, $D$ is the depth of the tree. The success probability is at least $0.99$. If the algorithm fails, then it returns FAIL.
\end{theorem}

\noindent\textbf{Conditional hardness for biconnectivity.}  A
conjectured hardness for the connectivity problem is the \emph{one
  cycle vs. two cycles} conjecture: for any $\gamma\geq 0$ and any
constant $\delta\in(0,1)$, any $(\gamma,\delta)$-MPC algorithm
requires $\Omega(\log n)$ parallel time to determine whether the input
$n$-vertex graph is a single cycle or contains two disjoint length
$n/2$ cycles.  This conjectured hardness result is widely used in the
MPC literature~\cite{ksv10,bks13,klmrv14,rvw16,yv18}.  Under this
conjecture, we show that $\Omega(\log D')$ parallel time is necessary
for the biconnectivity problem, and this is true even when $D=O(1)$,
i.e., the diameter of the graph is a constant.

\begin{theorem}[Hardness of biconnectivity in MPC]
	For any $\gamma\geq 0$ and any constant $\delta\in(0,1)$, unless 
	there is a $(\gamma,\delta)$-MPC algorithm which can distinguish the following two instances: 1) a single cycle with $n$ vertices, 2) two disjoint cycles each contains $n/2$ vertices, in $o(\log n)$ parallel time, any $(\gamma,\delta)$-MPC algorithm requires $\Omega( \log D')$ parallel time for testing whether a graph $G$  with a constant diameter is biconnected.  
\end{theorem}

\subsection{Our Techniques}

\noindent\textbf{Biconnectivity.} At a high level our biconnectivity algorithm is based on a framework proposed by~\cite{tv85}. 
The main idea is to construct a new graph and reduce the problem of finding biconnected components of $G$ to the problem of finding connected components of the new graph $G'$.
At first glance, it should be efficiently solved by the connectivity algorithm~\cite{asswz18}.
However, there are two main issues: 1) since the parallel time of the MPC connectivity algorithm of~\cite{asswz18} depends on the diameter of the input graph, we need to make the diameter of $G'$ small, 2) we need to construct $G'$ efficiently.
Let us first consider the first issue, and we will discuss the second issue later. 

We give an analysis of the diameter of $G'=(V',E')$ constructed by~\cite{tv85}. 
Without loss of generality, we  can suppose the input $G=(V,E)$ is connected.
Each vertex in $G'$ corresponds to an edge of $G$.
Let $T$ be an arbitrary spanning tree of $G$ with depth $d$.
Each non-tree edge $e$ can define a simple cycle $C_e$ which contains the edge $e$ and the unique path between the endpoints of $e$ in the tree $T$.
Thus, the length of $C_e$ is at most $2d+1$.
If there is a such cycle containing any two tree edges $(u,v),(v,w)$, vertices $(u,v),(v,w)$ are connected in $G'$. 
For each non-tree edge $e$, we connect the vertex $e$ to the vertex $e'$ in graph $G'$ where $e'$ is an arbitrary tree edge in the cycle $C_e$.
By the construction of $G'$, any $e,e'$ from the same connected components of $G'$ should be in the same biconnected components of $G$.
Now consider arbitrary two edges $e,e'$ in the same biconnected component of $G$.
There must be a simple cycle $C$ which contains both edges $e,e'$ in $G$.
Since all the simple cycles defined by the non-tree edges are a cycle basis of $G$~\cite{d18}, the edge set of $C$ can be represented by the xor sum of all the edge sets of $k$ basis cycles $C_{1},C_{2},\cdots,C_k$ where $C_i$ is a simple cycle defined by a non-tree edge $e_i$ on the cycle $C$. 
$k$ is upper bounded by the bi-diameter of $G$.
Furthermore, we can assume $C_{i}$ intersects $C_{i+1}$.
There should be a path between $e,e'$ in $G'$, and the length of the path is at most $\sum_{i=1}^k |C_i|\leq O(k\cdot d)$.
So, the diameter of $G'$ is upper bounded by $O(k\cdot d)$.
Thus, according to~\cite{asswz18}, we can find the connected components of $G'$ in $\sim(\log k+\log d)$ parallel time, where $d$ and $k$ are upper bounded by the diameter and the bi-diameter of $G$ respectively.

Now let us consider how to construct $G'$ efficiently.
The bottleneck is to determine whether the tree edges $(u,v),(v,w)$ should be connected in $G'$ or not.
Suppose $w$ is the parent of $v$ and $v$ is the parent of $u$.
The vertex $(u,v)$ should connect to the vertex $(v,w)$ in $G'$ if and only if there is a non-tree edge that connects a vertex $x$ in the subtree of $u$ and a vertex $y$ which is on the outside of the subtree of $v$.
For each vertex $x$, let $\bac(x)$ be the minimum depth of the least common ancestor (LCA) of $(x,y)$ over all the non-tree edges $(x,y)$.
Then $(u,v)$ should be connected to $(v,w)$ in $G'$ if and only if there is a vertex $x$ in the subtree of $u$ in $G$ such that $\bac(x)$ is smaller than the depth of $v$.
Since the vertices in a subtree should appear consecutively in the DFS sequence, this question can be solved by some range queries over the DFS sequence.
Next, we will discuss how to compute the DFS sequence of a tree.

\noindent\textbf{DFS sequence.} 
The DFS sequence of a tree is a variant of the Euler tour representation of the tree. 
For an $n$-vertex tree $T$,
\cite{tv85} gives an $O(\log n)$ parallel time PRAM algorithm for the Euler tour representation of $T$.
However, since their construction method will destroy the tree structure, it is hard to get a faster MPC algorithm based on this framework.
Instead, we follow the leaf sampling framework proposed by~\cite{asswz18}. 
Although the DFS sequence algorithm proposed by~\cite{asswz18} takes $O(\log d)$ time where $d$ is the depth of $T$, it needs $\Omega(n\log d)$ total space.
The bottleneck is the subroutine which needs to solve the least common ancestors problem and generate multiple path sequences.
The previous algorithm uses the doubling algorithm for the subroutine, i,e., for each vertex $v$, they store the $2^{i}$-th ancestor of $v$ for every $i\in[\lceil \log d \rceil]$. 
This is the reason why~\cite{asswz18} cannot achieve the linear total space.
We show how to compress the tree $T$ into a new tree $T'$ which only contains at most $n/\lceil\log d\rceil$ vertices. 
We argue that applying the doubling algorithm on $T'$ is sufficient for us to find the DFS sequence of $T$.

\noindent\textbf{$2$-Edge connectivity.} Without loss of generality, we can assume the input graph $G$ is connected.
Consider a rooted spanning tree $T$ and an edge $e=(u,v)$ in $G$.
Suppose the depth of $u$ is at least the depth of $v$ in $T$, i.e., $v$ cannot be a child of $u$.
The edge $e$ is not a bridge if and only if either $e$ is a non-tree edge or there is a non-tree edge $(x,y)$ connecting the subtree of $u$ and a vertex on the outside of the subtree of $u$.
Similarly, the second case can be solved by some range queries over the DFS sequence of $T$.

\noindent\textbf{Conditional hardness for biconnectivity.} We want to reduce the connectivity problem to the biconnectivity problem. 
For an undirected graph $G$, if we add an additional vertex $v^*$ and connects $v^*$ to every vertex of $G$, then the diameter of the resulting graph $G'$ is at most $2$ and each biconnected components of $G'$ corresponds to a connected component of $G$.
Furthermore, the bi-diameter of $G'$ is upper bounded by the diameter of $G$ plus $2$.
Therefore, if the parallel time of an algorithm $\mathcal{A}'$ for finding the biconnected components of $G'$ depends on the bi-diameter of $G'$, there exists an algorithm $\mathcal{A}$ which can find all the connected components of $G$ in the parallel time which has the same dependence on the diameter of $G$.


\subsection{A Roadmap}
The rest of this paper is organized as follows. Section~\ref{sec:preli} includes the notation and some useful definitions. 
Section~\ref{sec:two_edge_and_biconn} describes the offline algorithms for $2$-edge connectivity and biconnectivity.
It also includes the analysis of some crucial properties and the correctness of the algorithms.
In Section~\ref{sec:DFS_sequence_in_linear_total_space}, we show how to find the DFS sequence of a tree in the MPC model with linear total space.
Section~\ref{sec:mpc_biconnectivity} discusses the implementations of the $2$-edge connectivity algorithm and the biconnectivity algorithm in the MPC model.
Section~\ref{sec:hardness_of_biconnectivity} contains the conditional hardness result for the biconnectivity problem in the MPC model. 

\section{Preliminaries}\label{sec:preli}
We follow the notation of~\cite{asswz18}.
 $[n]$ denotes the set of integers $\{1,2,\cdots,n\}$.
 
\noindent\textbf{Diameter and bi-diameter.} Consider an undirected graph $G$ with a vertex set $V$ and an edge set $E$. For any two vertices $u,v$, we use $\dist_G(u,v)$ to denote the distance between $u$ and $v$ in graph $G$. If $u,v$ are not in the same (connected) component of $G$, then $\dist_G(u,v)=\infty$. The diameter $\diam(G)$ of $G$ is the largest diameter of its connected components, i.e., $\diam(G)=\max_{u,v\in V:\dist_G(u,v)\not=\infty} \dist_G(u,v)$. 
 $(v_1,v_2,\cdots,v_k)\in V^k$ is a cycle of length $k-1$ if $v_1=v_k$ and $\forall i\in[k-1],(v_i,v_{i+1})\in E$.
 We say a cycle $(v_1,v_2,\cdots,v_k)$ is simple if $k\geq 4$ and each vertex only appears once in the cycle except $v_1\ (v_k)$. 
 Consider two different vertices $u,v\in V$.
 We use $\cyclen_G(u,v)$ to denote the minimum length of a simple cycle which contains both vertices $u$ and $v$.
  If there is no simple cycle which contains both $u$ and $v$, $\cyclen_G(u,v)=\infty$.
  $\cyclen_G(u,u)$ is defined as $0$.
   The bi-diameter of $G$, $\bidiam(G)$, is defined as $\max_{u,v\in V:\cyclen_G(u,v)\not=\infty}\cyclen_G(u,v)$.

\noindent\textbf{Representation of a rooted forest.} Let $V$ denote a
set of vertices. We represent a rooted forest in the same manner as~\cite{asswz18}.
Consider a mapping $\p: V\rightarrow V$. For
$i\in\mathbb{N}_{>0}$ and $v\in V$, we define $\p^{(i)}(v)$ as
$\p(\p^{(i-1)}(v))$, and $\p^{(0)}(v)$ is defined as $v$ itself. If
$\forall v\in V,\exists i>0$ such that $\p^{(i)}(v)=\p^{(i+1)}(v)$,
then we call $\p$ a set of parent pointers on $V$. For $v\in V$, if
$\p(v)=v$, then we say $v$ is a root of $\p$. Notice that $\p$
actually can represent a rooted forest, thus $\p$ can have more than
one root. The depth of $v\in V$, $\dep_{\p(v)}$ is the smallest
$i\in\mathbb{N}$ such that $\p^{(i)}(v)$ is the same as
$\p^{(i+1)}(v)$. The root of $v\in V$, $\p^{(\infty)}(v)$ is defined
as $\p^{(\dep_{\p}(v))}(v)$. The depth of $\p,$ $\dep(\p)$ is defined
as $\max_{v\in V}\dep_{\p}(v)$.

\noindent\textbf{Ancestor and path.} For two vertices $u,v\in V$, if $\exists i\in\mathbb{N}$ such that $u=\p^{(i)}(v),$ then $u$ is an ancestor of $v$ (in $\p$).
If $u$ is an ancestor of $v$, then the path $P(v,u)$ (in $\p$) from $v$ to $u$ is a sequence $(v,\p(v),\p^{(2)}(v),\cdots,u)$ and the path $P(u,v)$ is the reverse of $P(v,u)$, i.e., $P(u,v)=(u,\cdots,\p^{(2)}(v),\p(v),v)$.
 If an ancestor $u$ of $v$ is also an ancestor of $w$, then $u$ is a common ancestor of $(v,w)$. Furthermore, if a common ancestor $u$ of $(v,w)$ satisfies $\dep_{\p}(u)\geq \dep_{\p}(x)$ for any common ancestor $x$ of $(v,w)$, then $u$ is the lowest common ancestor (LCA) of
$(v,w)$. 

\noindent\textbf{Children and leaves.} For any non-root vertex $u$ of $\p$, $u$ is a child of $\p(u)$. For any vertex $v\in V$, $\child_{\p}(v)$ denotes the set of all the children of $v$, i.e., $\child_{\p}(v)=\{u\in V\mid u\not=v,\p(u)=v\}.$ If $u$ is the $k^{\text{th}}$ smallest vertex in the set $\child_{\p}(v),$ then we define $\rank_{\p}(u)=k$, or in other words, $u$ is the $k^{\text{th}}$ child of $v$. If $v$ is a root vertex of $\p$, then $\rank_{\p}(v)$ is defined as $1$. $\child_{\p}(v,k)$ denotes the $k^{\text{th}}$ child of $v$.
For simplicity, if $\p$ is clear in the context, we just use $\child(v)$, $\rank(v)$ and $\child(v,k)$ to denote $\child_{\p}(v)$, $\rank_{\p}(v)$ and $\child_{\p}(v,k)$ for short. If $\child(v)=\emptyset$, then $v$ is a leaf of $\p$. We denote $\leaves(\p)$ as the set of all the leaves of $\p$, i.e., $\leaves(\p)=\{v\mid\child(v)=\emptyset\}$.

\subsection{Depth-First-Search Sequence}
The Euler tour representation of a tree is proposed by~\cite{tv84,tv85}. 
It is a crucial building block in many graph algorithms including biconnectivity algorithms.
The Depth-First-Search (DFS) sequence~\cite{asswz18} of a rooted tree is a variant of the Euler tour representation. Let us first introduce some relevant concepts of the DFS sequence.
\begin{definition}[Subtree~\cite{asswz18}]
	Consider a set of parent pointers $\p:V\rightarrow V$ on a vertex set $V$. Let $v$ be a vertex in $V$, and let $V'=\{u\in V\mid v\text{ is an ancestor of }u\}$. $\p':V'\rightarrow V'$ is a set of parent pointers on $V'$. If $\forall u\in V'\setminus \{v\}$, $\p'(u)=\p(u)$ and $\p'(v)=v$, then $\p'$ is a subtree of $v$ in $\p$. For $u\in V'$, we say $u$ is in the subtree of $v$.
\end{definition}
The definition of the DFS sequence is the following:
\begin{definition}[DFS sequence~\cite{asswz18}]\label{def:dfs_sequence}
	Consider a set of parent pointers $\p:V\rightarrow V$ on a vertex set $V$. Let $v$ be a vertex in $V$. If $v$ is a leaf in $\p$, then the DFS sequence of the subtree of $v$ is $(v)$. Otherwise, the DFS sequence of the subtree of $v$ is defined recursively as
	\begin{align*}
	(v,a_{1,1},a_{1,2},\cdots,a_{1,n_1},v,a_{2,1},a_{2,2},\cdots,a_{2,n_2},v,\cdots,a_{k,1},a_{k,2},\cdots,a_{k,n_k},v),
	\end{align*}
	where $k=|\child(v)|$ and $\forall i\in[k],$ $(a_{i,1},a_{i,2},\cdots,a_{i,n_i})$ is the DFS sequence of the subtree of $\child(v,i)$, i.e., the $i^{\text{th}}$ child of $v$.
\end{definition}
If $\p:V\rightarrow V$ has a unique root $v$, then we define the DFS sequence of $\p$ as the DFS sequence of the subtree of $v$. 
By the definition of the DFS sequence, for any two consecutive elements $a_i$ and $a_{i+1}$ in the sequence, $a_i$ is either a parent of $a_{i+1}$ or $a_i$ is a child of $a_{i+1}$. 
Furthermore, for any vertex $v$, if both elements $a_i$ and $a_j$ $(i<j)$ in the DFS sequence $A$ are $v$, any element $a_k$ between $a_i$ and $a_j$ (i.e., $i\leq k\leq j$) should be a vertex in the subtree of $v$. 

\subsection{Data Organization and Basic Algorithms in the MPC Model}\label{sec:data_organize}

We organize the data in the MPC model as in~\cite{asswz18}.

\noindent\textbf{Set.} Consider a set of $m$ items $S=\{x_1,x_2,\cdots,x_m\}$ where each $x_i$ can be described by a constant number of words. If $x\in S$ $\Leftrightarrow$ there is a unique machine which stores a pair $(``S",x)$ in its local memory, then the set $S$ is stored in the system. $``S"$ is the name of the set $S$ and can be represented by a constant number of words. Let $\mathcal{S}=\{S_1,S_2,\cdots,S_m\}$ be a family of sets, where $\forall i\in[m],S_i$ is stored in the system and the name of $S_i$ can be represented by a constant number of words. If $S\in \mathcal{S}$ $\Leftrightarrow$ there is a unique machine which stores a pair $(``\mathcal{S}",``S")$ in its local memory, then we say $\mathcal{S}$ is stored in the system. The  total space for storing $S$ is $\Theta(|S|)$.

An undirected graph $G$ can be represented by a pair of the sets $(V,E)$, where $V=\{v_1,v_2,\cdots,v_n\}$ denotes the set of the vertices and $E=\{(u_1,v_1),(u_2,v_2),\cdots,(u_m,v_m)\}\subseteq V\times V$ denotes the set of the edges. To store the graph $G$ in the system, we just need to store both $V$ and $E$ in the system.

\noindent\textbf{Mapping.} Consider a mapping $f:A\rightarrow B$ where $A,B$ are two finite sets and every element from $A$ or $B$ only requires a constant number of words to describe. 
Let $S=\{(a,b)\mid a\in A,b=f(a)\}$.
Then $S$ is a set representation of the mapping $f$, and the name of $S$ is $``f"$. If the set $S$ is stored in the system, then we say the mapping $f$ is stored in the system. The total space needed for storing $f$ is $\Theta(|A|)$.

A set of parent pointers on a vertex set $V$ can be regarded as a mapping $\p:V\rightarrow V$.

\noindent\textbf{Sequence.} Let $A=(a_1,a_2,\cdots,a_m)$ be a sequence of $m$ elements, where each element $a_i$ can be represented by a constant number of words. Let $S=\{(x_1,a_1),(x_2,a_2),\cdots,(x_m,a_m)\}$ where $x_1<x_2<\cdots<x_m\in \mathbb{R}$. 
Then $S$ is a set  representation of the sequence $A$, and the name of $S$ is $``A"$. If $S$ is stored in the system, then we say the sequence $A$ is stored in the system. The total space needed for storing $A$ is $\Theta(m)$.

\noindent\textbf{Basic MPC operations.} One of the most basic algorithm in the MPC model is sorting. 
\begin{theorem}[\cite{g99,gsz11}]\label{thm:MPC_sorting}
	Sorting can be solved in $c/\delta$ parallel time in the $(0,\delta)$-MPC model for any constant $\delta\in(0,1)$, where $c\geq 0$ is a universal constant.
\end{theorem}

For any $\delta'\geq \delta,$ $O(n^{\delta'-\delta})$ number of machines with $\Theta(n^{\delta})$ local memory can always be simulated by $O(1)$ number of machines with $\Theta(n^{\delta'})$ local memory. Therefore, if an algorithm can solve a problem in $(\gamma,\delta)$-MPC model in $R(n)$ rounds, then the such algorithm can be simulated in $(\gamma',\delta')$-MPC model in $O(R(n))$ rounds for any $\gamma'\geq \gamma,\delta'\geq \delta$. 
Thus, for any $\gamma\geq 0$ and any constant $\delta\in(0,1)$, sorting takes $O(1)$ parallel time in the $(\gamma,\delta)$-MPC model.


Sorting is an important tool to build the MPC subroutines. One such MPC subroutine is to handle multiple queries at the same time. Roughly speaking, a random access shared memory can be simulated in the MPC model. Suppose there are $k$ sets $S_1,S_2,\cdots,S_k$ stored in the system, and the $t$ of them are set representations of mappings $f_1:A_1\rightarrow B_1,f_2:A_2\rightarrow B_2,\cdots,f_t:A_t\rightarrow B_t$. Suppose each machine has several queries where each query requires the value $f_i(a)$ for some $i\in[t],a\in A_i$. All the queries can be simultaneously handled in constant parallel time in the $(0,\delta)$-MPC model for any constant $\delta\in(0,1)$. 
For more basic MPC operations, we refer readers to~\cite{asswz18}.

\section{$2$-Edge Connectivity and Biconnectivity}\label{sec:two_edge_and_biconn}
Consider a connected undirected graph $G$ with a vertex set $V$ and an edge set $E$. In the $2$-edge connectivity problem, the goal is to find all the bridges of $G$, where an edge $e\in E$ is called a bridge if its removal disconnects $G$. In the biconnectivity problem, the goal is to partition the edges into several groups $E_1,E_2,\cdots,E_k$, i.e., $E=\bigcup_{i=1}^k E_i,\forall i\not=j,E_i\cap E_j = \emptyset$, such that $\forall e\not =e'\in E$, $e$ and $e'$ are in the same group if and only if there is a simple cycle in $G$ which contains both $e$ and $e'$. A subgraph induced by an edge group $E_i$ is called a biconnected component (block). In other words, the goal of the biconnectivity problem is to find all the blocks of $G$.

In this section, we describe the algorithms for both the $2$-edge connectivity problem and the biconnectivity problem in the offline setting. In Section~\ref{sec:mpc_biconnectivity}, we will discuss how to implement them in the MPC model.

\subsection{$2$-Edge Connectivity}\label{sec:two_edge_connectivity}
The $2$-edge connectivity problem is much simpler than the biconnectivity problem. 
We first compute a spanning tree of the graph. Only a tree edge can be a bridge. Then for any non-root vertex $v$, if  there is no non-tree edge which crosses between the subtree of $v$ and the outside of the subtree of $v$, then the tree edge which connects $v$ to its parent is a bridge. 
\begin{figure}
\noindent\fbox{
	\begin{minipage}{\linewidth}
		$2$-Edge Connectivity Algorithm:
		\begin{itemize}
			\item \textbf{Input}:
			\begin{itemize}
				
				\item A connected undirected graph $G=(V,E)$.
				
			\end{itemize}
			
			\item \textbf{Output}:
			\begin{itemize}
				
				\item A subset of edges $B\subseteq E$.
				
			\end{itemize}
			
			\item \textbf{Finding bridges} (\textsc{Bridges}$(G=(V,E))$ ): 
			\begin{enumerate}
				\item Compute a rooted spanning tree of $G$. The spanning tree is represented by a set of parent pointers $\p:V\rightarrow V$. 
				\item Compute $\bac:V\rightarrow \mathbb{Z}_{\geq 0}$: for 
				each $v\in V,$ 
				$$\bac(v)\gets \min\left(\dep_{\p}(v), \min_{w\in V\setminus\{\p(v)\}: (v,w)\in E}\dep_{\p}(\text{the LCA of }(v,w))\right).$$\label{sta:bridge_set_bac}
				\item Compute the DFS sequence $A$ of $\p$.\label{sta:two_edge_dfs_sequence}
				\item Initialize $B\gets \emptyset.$ For each non-root vertex $v$, let $a_i,a_j$ be the first and the last appearance of $v$ in $A$ respectively. If $\min_{k:i\leq k\leq j}\bac(a_k)\geq \dep_{\p}(v)$, $B\leftarrow B\cup\{(v,\p(v))\}$. Output $B$. \label{sta:final_output_bridge}
			\end{enumerate}
		\end{itemize}
	\end{minipage}
}
\end{figure}
\begin{lemma}[$2$-Edge connectivity]\label{lem:two_edge_connect}
	Consider an undirected graph $G=(V,E)$. Let $B$ be the output of \textsc{Bridges}$(G)$. Then $B$ is the set of all the bridges of $G$.
\end{lemma}
\begin{proof}
Suppose $(u,v)\in E$ is not a bridge. If $(u,v)$ is a non-tree edge in $\p$, then since $B$ only contains tree edges, $(u,v)\not\in B$. Otherwise, suppose $\p(v)=u$. There must be a non-tree edge $(x,y)\in E$ such that $x$ is in the subtree of $v$ but $y$ is not in the subtree of $v$. Thus, the LCA of $(x,y)$ is not $v$, and it is an ancestor of $v$ which means that the depth of the LCA of $(x,y)$ is smaller than $\dep_{\p}(v)$. By step~\ref{sta:bridge_set_bac}, we have $\bac(x)<\dep_{\p}(v)$. 
Let $a_i,a_j$ be the first and the last appearance of $v$ in the DFS sequence of $\p$. Since $x$ is in the subtree of $v$, there exists $k\in \{i,i+1,\cdots,j\}$ such that $v=a_k$. By step~\ref{sta:final_output_bridge}, since $\bac(a_k)<\dep_{\p}(v)$, $(u,v)\not\in B$.

If $(u,v)\in E$ is a bridge. Then $(u,v)$ must be a tree edge in $\p$, i.e., either $\p(u)=v$ or $\p(v)=u$. Suppose $\p(v)=u$. Then for any non-tree edge $(x,y)$ with $x$ in the subtree of $v$, $y$ must also be in the subtree of $v$. Thus, the depth of the LCA of $(x,y)$ should be at least $\dep_{\p}(v)$. By step~\ref{sta:bridge_set_bac}, for any $x$ in the subtree of $v$, we have $\bac(x)\geq \dep_{\p}(v)$. Let $a_i,a_j$ be the first and the last appearance of $v$ in the DFS sequence of $\p$.
Since all the vertices $a_i,a_{i+1},\cdots,a_j$ are in the subtree of $v$, we have $(u,v)\in B$ by step~\ref{sta:final_output_bridge}.
\end{proof}

\subsection{Biconnectivity}\label{sec:biconn}
In this section, we will show a biconnectivity algorithm. 
It is a modification of the algorithm proposed by~\cite{tv85}. 
The high level idea is to construct a new graph $G'$ based on the input graph $G$, and reduce the biconnectivity problem of $G$ to the connectivity problem of $G'$.
Since the running time of the connectivity algorithm~\cite{asswz18} depends on the diameter of the graph, we also give an analysis of the diameter of the graph $G'$.
\begin{figure}
	\noindent\fbox{
		\begin{minipage}{\linewidth}
			Biconnectivity Algorithm:
			\begin{itemize}
				\item \textbf{Input}:
				\begin{itemize}
					
					\item A connected undirected graph $G=(V,E)$.
					
				\end{itemize}
				
				\item \textbf{Output}:
				\begin{itemize}
					
					\item A coloring $\col: E\rightarrow V$ of the edges.
					
				\end{itemize}
				
				\item \textbf{Finding blocks} (\textsc{Biconn}$(G=(V,E))$ ): 
				\begin{enumerate}
					\item Compute a rooted spanning tree of $G$. The spanning tree is represented by a set of parent pointers $\p:V\rightarrow V$. 
					\item Compute $\bac:V\rightarrow \mathbb{Z}_{\geq 0}$: for 
					each $v\in V,$ 
					$$\bac(v)\gets \min\left(\dep_{\p}(v), \min_{w\in V\setminus\{\p(v)\}: (v,w)\in E}\dep_{\p}(\text{the LCA of }(v,w))\right).$$\label{sta:block_set_bac}
					\item Compute the DFS sequence $A$ of $\p$.
					\item Let $r$ be the root of $\p$. Initialize $V'\gets V\setminus\{r\},E'\gets\emptyset$.
					\item For each $v\in V'$, let $a_i,a_j$ be the first and the last appearance of $v$ in $A$ respectively. If $\min_{k\in\{i,i+1,\cdots,j\}}\bac(a_k)< \dep_{\p}(\p(v))$, $E'\leftarrow E'\cup\{(v,\p(v))\}$. \label{sta:non_tree_edge_ancestor}
					\item For each $(u,v)\in E$, if neither $u$ nor $v$ is the LCA of $(u,v)$ in $\p$, $E'\gets E'\cup\{(u,v)\}$. \label{sta:non_tree_edge_non_ancestor}
					\item Compute the connected components of $G'=(V',E')$. Let $\col':V'\rightarrow V'$ be the coloring of the vertices in $V'$ such that $\forall u',v'\in V'$, $u',v'$ are in the same connected component in $G'$ $\Leftrightarrow$ $\col'(u')=\col'(v')$. \label{sta:connected_components_Gprime}
					\item Initialize $\col:E\rightarrow V$. For each $e=(u,v)\in E$, if $\dep_{\p}(u)\geq \dep_{\p}(v)$, set $\col(e)\gets \col'(u)$; otherwise, set $\col(e)\gets \col'(v)$. Output $\col:E\rightarrow V$. \label{sta:final_assignment}
				\end{enumerate}
			\end{itemize}
		\end{minipage}
	}
\end{figure}
\begin{lemma}[Biconnectivity]\label{lem:biconnectivity}
	Consider an undirected graph $G=(V,E)$. Let $\col:E\rightarrow V$ be the output of \textsc{Biconn}$(G)$. Then $\forall e,e'\in E,e\not=e',$ $\col$ satisfies $\col(e)=\col(e')$ $\Leftrightarrow$ there is a simple cycle in $G$ which contains both $e$ and $e'$. Furthermore, the diameter of the graph $G'$ constructed by \textsc{Biconn}$(G)$ is at most $O(\dep(\p)\cdot \bidiam(G))$, the number of vertices of $G'$ is at most $|V|$, and the number of edges of $G'$ is at most $|E|$.
\end{lemma}

\begin{proof}
Each $v\in V'$ corresponds to a tree edge $(\p(v),v)\in E$. 
Since $V'\subset V$, $|V'|\leq |V|$.
By step~\ref{sta:non_tree_edge_ancestor} and step~\ref{sta:non_tree_edge_non_ancestor}, each edge of $G$  creates at most $1$ edge of $G'$.
Thus, $|E'|\leq |E|$. 
\begin{claim}\label{cla:uv_implies_in_a_cycle}
If $\dist_{G'}(u,v)<\infty$, i.e., vertices $u,v\in V'$ are in the same connected component of $G'$, then there is a simple cycle in $G$ which contains both edges $(u,\p(u))$ and $(v,\p(v))$.
\end{claim}
\begin{proof}
Firstly, let us consider the case when $(u,v)\in E'$.
If $(u,v)$ is added into $E'$ by step~\ref{sta:non_tree_edge_non_ancestor}, then there is a simple cycle in $G$: 
$$
(u,\p^{(1)}(u),\p^{(2)}(u),\cdots,\text{the LCA of }(u,v),\cdots,\p^{(2)}(v),\p^{(1)}(v),v,u).
$$
Both edges $(u,\p(u))$ and $(v,\p(v))$ are in the such cycle.
If $(u,v)$ is added into $E'$ by step~\ref{sta:non_tree_edge_ancestor}, then $u=\p(v)$. Let $a_i,a_j$ be the first and the last appearance of $v$ in $A$ respectively. By step~\ref{sta:non_tree_edge_ancestor}, there exists $k$ with $i\leq k\leq j$ such that $\bac(a_k)<\dep_{\p}(v)$. Thus, there is a vertex $x$ in the subtree of $v$ such that $\bac(x)<\dep_{\p}(u)$. By step~\ref{sta:block_set_bac}, there is an edge $(x,y)\in E$ such that the depth of the LCA of $(x,y)$ is smaller than $\dep_{\p}(u)$ which means that $y$ is not in the subtree of $u$. In this case, there is a simple cycle in $G$:
$$
(x,\p^{(1)}(x),\p^{(2)}(x),\cdots,v,u,\p(u),\cdots,\text{the LCA of }(x,y),\cdots,\p^{(2)}(y),\p^{(1)}(y),y,x).
$$
Since $u=\p(v)$, both edges $(v,\p(v))$, $(u,\p(u))$ are in the such cycle.

Suppose $v,u\in V'$ are in the same connected component of $G'$ and $(v,\p(v))$, $(u,\p(u))$ are in a simple cycle $C_1$ in $G$. 
Suppose $u,w\in V'$ are in the same connected component of $G'$ and $(u,\p(u))$, $(w,\p(w))$ are in a simple cycle $C_2$ in $G$. 
Then, $v$ and $w$ are in the same connected component of $G'$. 
The symmetric difference of the edge set of $C_1$ and the edge set of $C_2$ should form another simple cycle $C_3$ in $G$ which contains both edges $(v,\p(v))$ and $(w,\p(w))$.
By induction on $\dist_{G'}(v,w)$, the claim holds.
\end{proof}
 By Claim~\ref{cla:uv_implies_in_a_cycle} and step~\ref{sta:final_assignment}, $\forall u,v\in V'$, if $\col((u,\p(u)))=\col((v,\p(v)))$, then there should be a simple cycle in $G$ which contains both edges $(u,\p(u))$ and $(v,\p(v))$. 
Consider an edge $(u,v)\in E$ such that neither $u$ nor $v$ is the LCA of $(u,v)$, i.e., $(u,v)$ is a non-tree edge. 
Without loss of generality, suppose $\dep_{\p}(u)\geq \dep_{\p}(v)$.
There is always a cycle in $G$:
$$
(u,\p^{(1)}(u),\p^{(2)}(u),\cdots,\text{the LCA of }(u,v),\cdots,\p^{(2)}(v),\p^{(1)}(v),v,u),
$$
which contains both edges $(u,v),(u,\p(u))$.
By step~\ref{sta:final_assignment}, we have $\col((u,v))=\col((u,\p(u)))=\col'(u)$.
Therefore, $\forall e_1,e_2\in E$, there are always tree edges $e_1',e_2'\in E$ such that $\col(e_1')=\col(e_1),\col(e_2')=\col(e_2)$, $e_1,e_1'$ are either in a simple cycle in $G$ or $e_1=e_1'$, and $e_2,e_2'$ are either in a simple cycle in $G$ or $e_2=e_2'$.
If $\col(e_1)=\col(e_2)$, then $\col(e_1')=\col(e_2')$ which implies that $e_1',e_2'$ are either in a simple cycle in $G$ or $e_1'=e_2'$. Hence if $\col(e_1)=\col(e_2),$ then either there is a simple cycle in $G$ which contains both $e_1,e_2$ or $e_1=e_2$.

Next, let us show that if there is a simple cycle in $G$ which contains both edges $e,e'\in E$, then $\col(e)=\col(e')$. An observation is that each non-tree edge $e=(u,v)$ (i.e., neither $u$ nor $v$ is the LCA of $(u,v)$ in $\p$) defines a simple cycle $C_e$ in $G$:
$$
(u,\p^{(1)}(u),\cdots,\text{the LCA of }(u,v),\cdots,\p^{(1)}(v),v,u).
$$
\begin{claim}\label{cla:define_a_path_in_G_prime}
For any simple cycle $C_e$ defined by a non-tree edge $e=(u,v)$, there is a path $P_e$ in $G'$ such that $P_e$ contains every vertex in $C_e$ except the LCA of $(u,v)$ in $\p$. Furthermore, the length of $P_e$ is at most $2\dep(\p)$.
\end{claim}
\begin{proof}
Without loss of generality, we can assume $\dep_{\p}(u)\geq \dep_{\p}(v)$.
If $v$ is an ancestor of $u$, then the cycle $C_e$ is 
$$
(u,\p^{(1)}(u),\p^{(2)}(u),\cdots,\p^{(s)}(u),v,u)
$$
for some $s\geq 1$.
For each $j\in[s]$, $u$ is in the subtree of $\p^{(j-1)}(u)$. By step~\ref{sta:non_tree_edge_ancestor}, since $\bac(u)\leq \dep_{\p}(v)<\p^{(j)}(u)$ for any $j\in[s]$, we have $(\p^{(j-1)}(u),\p^{(j)}(u))\in E'$. Thus, there is a path $P_e$ in $G'$: $(u,\p^{(1)}(u),\p^{(2)}(u),\cdots,\p^{(s)}(u))$. In this case, the length of $P_e$ should be at most $\dep(\p)$.

If $v$ is not an ancestor of $u$, then the cycle $C_e$ is 
$$
(u,\p^{(1)}(u),\cdots,\p^{(s_1)}(u),\text{the LCA of }(u,v),\p^{(s_2)}(v),\cdots,\p^{(1)}(v),v,u)
$$
for some $s_1,s_2\geq 1$.
By the similar argument, $\forall j\in[s_1]$ the edge $(\p^{(j-1)}(u),\p^{(j)}(u))$ ($\forall j'\in [s_2]$ the edge $(\p^{(j'-1)}(v),\p^{(j')}(v))$) is added into $E'$ by step~\ref{sta:non_tree_edge_ancestor}.
By step~\ref{sta:non_tree_edge_non_ancestor}, $(u,v)$ is added into $E'$.
Therefore, there is a path $P_e$ in $G'$:
$$(\p^{(s_1)}(u),\p^{(s_1-1)}(u),\cdots,\p^{(1)}(u),u,v,\p^{(1)}(v),\p^{(2)}(v),\cdots,\p^{(s_2)}(v)).$$
In this case, the length of $P_e$ should be at most $2\dep(\p)-1$.
\end{proof}
Notice that all the simple cycles defined by the non-tree edges formed a cycle basis of the cycle space of $G$, i.e., the edge set of any simple cycle in $G$ can be represented by an xor sum of the edge sets of cycles $C_{e_1},C_{e_2},\cdots,C_{e_s}$ defined by some non-tree edges $e_1,e_2,\cdots,e_s \in E$~\cite{d18}.
Consider any two tree edges $(u,\p(u)),(v,\p(v))\in E$ contained by a simple cycle $C$. Let $e_1,e_2,\cdots,e_s\in E$ be all the non-tree edges in $C$. Then $C$ can be represented by an xor sum of $C_{e_1},C_{e_2},\cdots,C_{e_s}$. Furthermore, $\forall i\in[s-1],$ $C_{e_{i}}$ and $C_{e_{i+1}}$ should have a common tree edge. According to Claim~\ref{cla:define_a_path_in_G_prime}, for each $i\in[s]$, we can find a path $P_{e_i}$ in $G'$ and $\forall j\in[s-1]$, $P_{e_j}$ intersects $P_{e_{j+1}}$.
Therefore, $u$ and $v$ are in the same connected component in $G'$.
By step~\ref{sta:final_assignment}, $\col((u,\p(u)))=\col'(u)=\col'(v)=\col((v,\p(v)))$.
Now consider a non-tree edge $e=(u,v)\in E$. 
Without loss of generality, we can assume $\dep_{\p}(u)\geq \dep_{\p}(v)$.
A tree edge $(u,\p(u))$ is the simple cycle $C_e$ defined by $e$.
By step~\ref{sta:final_assignment}, we know that $\col(e)=\col'(u)=\col((u,\p(u)))$.
Therefore, we can conclude that $\forall e_1,e_2\in E$, if there is a simple cycle in $G$ which contains both $e_1,e_2$, then $\col(e_1)=\col(e_2)$.

The only thing remaining to prove is the diameter of $G'$. According to Claim~\ref{cla:uv_implies_in_a_cycle}, $\forall u,v\in V'$ with $\dist_{G'}(u,v)<\infty$, there is a cycle $C$ in $G$ which contains both edges $(u,\p(u))$ and $(v,\p(v))$. 

\begin{claim}
$\forall u,v\in V'$, if there is a cycle in $G$ which contains both edges $(u,\p(u))$, $(v,\p(v))$, then there is a cycle $C$ in $G$ with length $O(\bidiam(G))$ which contains both edges $(u,\p(u))$, $(v,\p(v))$.
\end{claim}
\begin{proof}
By the definition of $\bidiam(G)$, there is a cycle $C_1$ with length at most $\bidiam(G)$ which contains both vertices $u,v$. If $C_1$ already contains both edges $(u,\p(u))$, $(v,\p(v))$, then we are done. Otherwise, suppose $C_1$ does not contain $(u,p(u))$. There is an another cycle $C_2$ with length at most $\bidiam(G)$ which contains both vertices $\p(u),v$. We can regard $C_2$ as two disjoint paths from $\p(u)$ to $v$. Thus at least one of the path does not contain the edge $(u,\p(u))$. Suppose this path is $(\p(u),\cdots,x,\cdots,v)$ where $x$ is the first vertex which appears in $C_1$, then we can combine the path $(u,\p(u),\cdots,x)$ with the path obtained by removing the sub-path from $u$ to $x$ of $C_1$ to get a new cycle which contains both the edge $(u,\p(u))$ and $v$. The length of the new cycle is at most $2\cdot \bidiam(G)$. We can do the similar operation to add edge $(v,\p(v))$ into the cycle. Thus, finally we will get a cycle which contains both  $(u,\p(u))$, $(v,\p(v))$ with length at most $3\cdot\bidiam(G)$.
\end{proof}
According to the above claim, we can find a cycle $C$ in $G$ which contains both edges $(u,\p(u))$, $(v,\p(v))$ with length at most $O(\bidiam(G))$. It means that $C$ can be represented by an xor sum of $s\leq O(\bidiam(G))$ basis cycles $C_{e_1},C_{e_2},\cdots,C_{e_s}$ defined by non-tree edges $e_1,e_2,\cdots,e_s$. Furthermore, $\forall i\in[s-1],$ $C_{i}$ and $C_{i+1}$ have at least one common tree edge. By Claim~\ref{cla:define_a_path_in_G_prime}, we can find $s$ paths $P_{e_1},P_{e_2},\cdots,P_{e_s}$ defined by $e_1,e_2,\cdots,e_s$ in $G'$ such that $\forall i\in[s-1],$ $P_{e_i}$ intersects $P_{e_{i+1}}$ at some vertex, and $u,v$ are on some path $P_{e_x},P_{e_y}$ respectively. 
Thus, $\dist_{G'}(u,v)\leq \sum_{i=1}^s |P_{e_i}|\leq s\cdot O(\dep(\p))\leq O(\dep(\p)\cdot \bidiam(G))$, where the second inequality follows from Claim~\ref{cla:define_a_path_in_G_prime}.
To conclude, $\diam(G')\leq O(\dep(\p)\cdot \bidiam(G))$.
\end{proof}

\section{Parallel DFS Sequence in Linear Total Space}\label{sec:DFS_sequence_in_linear_total_space}

In Section~\ref{sec:leaf_sampling}, we will review an algorithmic framework proposed by~\cite{asswz18} for the DFS sequence.
In Section~\ref{sec:compressed_rooted_tree},~\ref{sec:LCA},~\ref{sec:path_generation}, we will discuss the subroutines needed for our DFS sequence algorithm in the offline setting. In Section~\ref{sec:implement_DFS}, we will discuss the implementation in the MPC model.  

\subsection{DFS Sequence via Leaf Sampling}\label{sec:leaf_sampling}

In the following, we review the leaf sampling algorithmic framework proposed by~\cite{asswz18} for finding the DFS sequence of a rooted tree.

\begin{figure}[t!]
\noindent\fbox{
	\begin{minipage}{\linewidth}
		Leaf Sampling Algorithm for DFS Sequence:
		\begin{itemize}
			\item \textbf{Pre-determined}:
			\begin{itemize}
				\item A threshold value $s$.\hfill{//$s$ will be the local memory size in the MPC model.}
			\end{itemize}
			\item \textbf{Input}:
			\begin{itemize}
				
				\item A rooted tree represented by a set of parent pointers $\p:V\rightarrow V$ on a set $V$ of $n$ vertices (i.e., $\p$ has a unique root $r$).
				
			\end{itemize}
			
			\item \textbf{Output}:
			\begin{itemize}
				
				\item The DFS sequence of the rooted tree represented by $\p$.
				
			\end{itemize}
			
			\item \textbf{Leaf sampling algorithm} (\textsc{LeafSampling}$(s,\p:V\rightarrow V)$ ): 
			\begin{enumerate}
				\item If $n\leq s$, return the DFS sequence of $\p$ directly.
				\item Set $t\gets \Theta(s^{1/3}\log n)$, $L\gets \leaves(\p)$.
				\item Each $v\in L$ is independently chosen with probability $p=\min(1,t/|L|)$, and let $S=\{l_1,l_2,\cdots,l_k\}$ be the set of samples. If $|S|^2 > s$, output FAIL.
				\item For every pair of sampled leaves $x,y\in S$ with $x\not =y$, find the least common ancestor $p_{x,y}$ of $(x,y)$, and set $p_{xy,x}$, $p_{xy,y}$ to be two children of $p_{x,y}$ such that $p_{xy,x}$ is an ancestor of $x$ and $p_{xy,y}$ is an ancestor of $y$. 
				\label{sta:finding_LCA}
				\item Sort $l_1,l_2,\cdots,l_k\in S$ such that $\forall i<j\in[k]$, $\rank(p_{l_il_j,l_i})<\rank(p_{l_il_j,l_j})$.
				\item Find the paths $A'_1=P(r,l_1),A'_2=P(\p(l_1),p_{l_1,l_2}),A'_3=P(p_{l_1l_2,l_2},l_2),\cdots,A'_{2k-2}=P(\p(l_{k-1}),p_{l_{k-1},l_k}),A'_{2k-1}=P(p_{l_{k-1}l_{k},l_k},l_k),A'_{2k}=P(l_{2k},r)$, i.e., the paths: $r\rightarrow l_1\rightarrow \text{the LCA of }(l_1,l_2)\rightarrow l_2\rightarrow\cdots\rightarrow l_{k-1}\rightarrow \text{the LCA of }(l_{k-1},l_k)\rightarrow l_k\rightarrow r$.
				\label{sta:finding_paths}
				\item Set $A'\gets A'_1A'_2\cdots A'_{2k},$ i.e., $A'$ is the concatenation of $A'_1,A'_2,\cdots,A'_{2k}$. 
				\item For each element $a'_i$ in the $i^{\text{th}}$ ($i>1$) position of the sequence $A'$,
				\begin{itemize}
					\item if the vertex $a'_i$ is a leaf, keep $a'_i$ as a single copy;
					\item Otherwise, 
					\begin{itemize}
						\item if $a'_{i-1}=\p(a'_i)$, i.e., $i$ is the first position that the vertex $a'_i$ appears in $A'$, split $a'_i$ into $\rank(a'_{i+1})$ copies; \hfill{//$a'_{i+1}$ is a child of $a'_i$.}
						\item if  $a'_{i-1},a'_{i+1}\in\child(a'_i)$, split $a'_i$ into $\rank(a'_{i+1})-\rank(a'_{i-1})$ copies;
						\item if $a'_{i+1}=\p(a'_i)$, i.e., $i$ is the last position that the vertex $a'_i$ appears in $A'$, split $a'_i$ into $|\child(a'_i)|-\rank(a'_{i-1})$ copies. \hfill{//$a'_{i-1}$ is a child of $a'_i$.}
					\end{itemize}
				\end{itemize}
				Let $A''$ be the result sequence.
				\item For each $v\in V$, if $\p(v)$ appears in $A''$ but $v$ does not appear in $A''$, recursively find the DFS sequence of the subtree of $v$, and insert the such sequence into the position after the $\rank(v)^{\text{th}}$ appearance of $\p(v)$ in $A''$.
				Output the final result sequence $A$. 
			\end{enumerate}
		\end{itemize}
	\end{minipage}
}
\end{figure}
\begin{theorem}[Leaf sampling algorithm~\cite{asswz18}]\label{thm:previous_leaf_sampling}
	Consider a set of parent pointers $\p:V\rightarrow V$ on a set $V$ of $n$ vertices. Suppose $\p$ has a unique root. For any $\gamma\geq 0$ and any constant $\delta\in(0,1)$, if both of step~\ref{sta:finding_LCA} and step~\ref{sta:finding_paths} in \textsc{LeafSampling$(n^{\delta},\p)$} can be implemented in the $(\gamma,\delta)$-MPC model with $O(\log(\dep(\p)))$ parallel time, then the leaf sampling algorithm with parameter $s=n^\delta$ on input $\p:V\rightarrow V$ can be implemented in the $(\gamma,\delta)$-MPC model. Furthermore, with probability at least $0.99$, \textsc{LeafSampling$(n^{\delta},\p)$} can output the DFS sequence of $\p$ in $O(\log(\dep(\p)))$ parallel time. If the algorithm fails, then it returns FAIL.
\end{theorem}

By Theorem~\ref{thm:previous_leaf_sampling}, we only need to give a linear total space MPC algorithm for the LCA problem and the path generation problem to design an efficient DFS sequence algorithm in the $(0,\delta)$-MPC model. 

In \cite{asswz18}, they proposed to use doubling algorithms to compute the LCA and generate the paths. Since they need to store the every $2^i$-th ancestor for each vertex, the total space needed is $\Theta(n\cdot\log(\text{the depth of the tree}))$. We will show that we only need to apply the doubling algorithm for a compressed tree, instead of applying the doubling algorithm for the original tree.

\subsection{Compressed Rooted Tree}\label{sec:compressed_rooted_tree}
Given a set of parent pointers $\p:V\rightarrow V$, we will show how to compress the rooted tree represented by $\p$.
\begin{figure}
\noindent\fbox{
	\begin{minipage}{\linewidth}
		Construction of a Compressed Rooted Tree:
		\begin{itemize}
			\item \textbf{Input}:
			\begin{itemize}
				
				\item A rooted tree represented by a set of parent pointers $\p:V\rightarrow V$ on a set $V$ of $n$ vertices ($\p$ has a unique root $r$).
				
			\end{itemize}
			
			\item \textbf{Output}:
			\begin{itemize}
				
				\item A vertex set $V'\subseteq V$, a set of parent pointers $\p':V'\rightarrow V'$ on $V'$.
				
			\end{itemize}
			
			\item \textbf{Tree compression} (\textsc{Compress}$(\p:V\rightarrow V)$ ): 
			\begin{enumerate}
				\item Compute the depth of $\p$, the depth of each vertex and set $d\gets \dep(\p)$, $t\gets \lceil\log d\rceil$.
				\item $V'\gets \{v\in V \mid \dep_{\p}(v)\mod t = 0, \dep_{\p}(v)+t\leq d\}$.
				\item Initialize $\p':V'\rightarrow V'$. For each $v\in V'$, $\p'(v)\gets \p^{(t)}(v)$.
				\item Output $V'$, $\p'$.
			\end{enumerate}
		\end{itemize}
	\end{minipage}
}
\end{figure}
\begin{lemma}[Properties of a compressed rooted tree]\label{lem:properties_compressed_rooted_tree}
Let $\p:V\rightarrow V$ be a set of parent pointers on a vertex set $V$ with $|V|> 1$, and $\p$ has a unique root. Let $t=\lceil\log(\dep(\p))\rceil$ and let $(V',\p')=$\textsc{Compress}$(\p)$. Then it has the following properties:
\begin{enumerate}
\item $|V'|\leq |V|/\log(\dep(\p))$.
\item $\forall v\in V',i\in \mathbb{N},$ $\p'^{(i)}(v)=\p^{(i\cdot t)}(v)\in V'$.
\item $\forall v\in V,$ $\exists i\in\{0,1,\cdots,2t\},$ such that $\p^{(i)}(v)\in V'$.
\end{enumerate}
\end{lemma}
\begin{proof}
Consider the first property. For each $v\in V'$, we define a set $$S(v)=\{u\in V\mid \dep_{\p}(u)>\dep_{\p}(v), \exists i\in[t-1], \p^{(i)}(u) = v \}.$$ 
$\forall u\in S(v),$ we have $\dep_{\p}(u)-\dep_{\p}(v)<t$.
Since $\forall v\in V'$, $\dep_{\p}(v)\mod t=0$, we have $S(v)\cap V'=\emptyset$. 
Furthermore, it is easy to show that $\forall u\not =v\in V'$, $S(u)\cap S(v)=\emptyset$.
Thus, $|V'|+\sum_{v\in V'} |S(v)| = \sum_{v\in V'} (|S(v)|+1) \leq |V|$. 
On the other hand, since $\forall v\in V'$, $\dep_{\p}(v)+t\leq \dep(\p),$ we know that $|S(v)|\geq t-1$.
Therefore $\sum_{v\in V'}(|S(v)|+1)\geq |V'|\cdot t$. 
To conclude, $|V'|\leq |V|/t\leq |V|/\log(\dep(\p))$.

Consider the second property.
If $v$ is a root vertex, $\p'(v)=\p^{(t)}(v)=v\in V'$. 
For a non-root vertex $v\in V'$, $\dep_{\p}(\p^{(t)}(v))=\dep_{\p}(v)-t$. Since $\dep_{\p}(v) \mod t = 0$, we have $\dep_{\p}(\p^{(t)}(v))\mod t=0$ which means that $\p'(v)=\p^{(t)}(v)\in V'$. Now we prove by induction. Suppose $\p'^{(i-1)}(v)=\p^{((i-1)\cdot t)}(v)$, then $\p'^{(i)}(v)=\p'(\p'^{(i-1)}(v))=\p^{(t)}(\p^{((i-1)\cdot t)}(v))=\p^{(i\cdot t)}(v)$.

Consider the third property. For $v\in V$, $\exists j\in\{0,1,\cdots,t-1\},$ such that $\dep_{\p}(\p^{(j)}(v))\mod t = 0$. Since $\dep_{\p}(\p^{(j+t)}(v))\mod t = 0$ and $\dep_{\p}(\p^{(j+t)}(v)) + t\leq \dep(\p)$, we know that $\p^{(j+t)}(v)\in V'$. Since $j+t\leq 2t$, the property holds.
\end{proof}

\subsection{Least Common Ancestor}\label{sec:LCA}
Given a rooted tree represented by a set of parent pointers $\p:V\rightarrow V$ on a vertex set $V$, and a set of $q$ queries $Q=\{(u_1,v_1),(u_2,v_2),\cdots,(u_q,v_q)\}$ where $\forall i\in[q],u_i\not =v_i,u_i,v_i\in \leaves(\p)$, we show a space efficient algorithm which can output the LCA of each queried pair of vertices. Notice that the assumption that queries only contain leaves is without loss of generality: we can attach an additional child vertex $v$ to each non-leaf vertex $u$. Thus, $v$ is a leaf vertex. When a query contains $u$, we can use $v$ to replace $u$ in the query, and the result will not change. 
\begin{figure}
\noindent\fbox{
	\begin{minipage}{\linewidth}
		Lowest Common Ancestor:
		\begin{itemize}
			\item \textbf{Input}:
			\begin{itemize}
				
				\item A rooted tree represented by a set of parent pointers $\p:V\rightarrow V$ on a set $V$ of $n$ vertices ($\p$ has a unique root $r$), and a set of $q$ queries $Q=\{(u_1,v_1),(u_2,v_2),\cdots,(u_q,v_q)\}$ where $\forall i\in[q],u_i\not =v_i,u_i,v_i\in \leaves(\p)$.
				
			\end{itemize}
			
			\item \textbf{Output}:
			\begin{itemize}
				
				\item $\lca:Q\rightarrow V\times V\times V$.
				
			\end{itemize}
			
			\item \textbf{Finding LCA} (\textsc{LCA}$(\p:V\rightarrow V,Q)$ ): 
			\begin{enumerate}
				\item $(V',\p')\gets$\textsc{Compress}$(\p)$. \hfill{//(see Lemma~\ref{lem:properties_compressed_rooted_tree}).}
				\item Set $d\gets \dep(\p), t\gets \lceil\log d\rceil$ and compute mappings $g_0,g_1,\cdots g_{t}:V'\rightarrow V'$ such that $\forall v\in V',j\in\{0,1,\cdots,t \}$, $g_j(v)=\p'^{(2^j)}(v)$. \label{sta:find_sparse_table}
				\item For each query $(u_i,v_i)\in Q$:\hfill{//Suppose $\dep_{\p}(u_i)\geq\dep_{\p}(v_i)$.}\label{sta:handle_each_lca_query}
				\begin{enumerate}
					\item If $\dep_{\p}(u_i)>\dep_{\p}(v_i) + 2t$, find an ancestor $\hat{u}_i$ of $u_i$ in $\p$ such that $\dep_{\p}(\hat{u}_i)\leq \dep_{\p}(v_i)+2t$ and $\dep_{\p}(\hat{u}_i)\geq\dep_{\p}(v_i)$. Otherwise, $\hat{u}_i\gets u_i$. \label{sta:find_a_close_u}
					\item If $\exists j\in[4t]$ $\p^{(j)}(\hat{u}_i)$ is the LCA of $(\hat{u}_i,v_i)$ in $\p$, set $\lca(u_i,v_i)=(\p^{(j)}(\hat{u}_i),x,y)$ where $x,y$ are children of $\p^{(j)}(\hat{u}_i)$ and $x,y$ are ancestors of $\hat{u}_i,v_i$ respectively. The query of $(u_i,v_i)$ is finished.
					\label{sta:check_directly}
					\item Find an ancestor $u'_i$ of $\hat{u}_i$ in $\p$ such that $u'_i$ is the closest vertex to $\hat{u}_i$ in $V'$, i.e., $\dep_{\p}(\hat{u}_i) - \dep_{\p}(u'_i)$ is minimized. Similarly, find an ancestor $v'_i$ of $v_i$ in $\p$ such that $v'_i$ is the closest vertex to $v_i$ in $V'$, i.e., $\dep_{\p}(v_i)-\dep_{\p}(v'_i)$ is minimized.\label{sta:find_vertex_in_Vprime}
					\item Find $u''_i\not=v''_i\in V'$ such that they are ancestors of $u'_i$ and $v'_i$ respectively, and $\p'(u''_i)=\p'(v''_i)$ is the LCA of $(u'_i,v'_i)$ in $\p'$.\label{sta:find_uv_double_prime}
					\item Find the smallest $j\in[2t]$ such that $\p^{(j)}(u''_i)=\p^{(j)}(v''_i)$. Set $\lca(u_i,v_i)=(\p^{(j)}(u''_i),\p^{(j-1)}(u''_i),\p^{(j-1)}(v''_i))$.\label{sta:find_final_LCA}
				\end{enumerate}
			\end{enumerate}
		\end{itemize}
	\end{minipage}
}
\end{figure}

Before we analyze the algorithm \textsc{LCA}$(\p,Q)$, let us discuss some details of the algorithm.
\begin{enumerate}
\item We pre-compute $\dep_{\p}(v)$ and $\dep_{\p'}(u)$ for every $v\in V$ and $u\in V'$.
\item To implement step~\ref{sta:find_a_close_u}, we firstly check whether $\dep_{\p}(u_i)>\dep_{\p}(v_i)+2t$. If it is not true, we can set $\hat{u}_i$ to be $u_i$ directly. Otherwise, according to Lemma~\ref{lem:properties_compressed_rooted_tree}, there is a $j\in\{0,1,\cdots,2t\}$ such that $\p^{(j)}(u_i)\in V'$. Since $\dep_{\p}(u_i)>\dep_{\p}(v_i)+2t$, $\dep_{\p}(\p^{(j)}(u_i))>\dep_{\p}(v_i)$. We initialize $\hat{u}_i$ to be $\p^{(j)}(u_i)\in V'$. For $k=t\rightarrow 0$, if $\dep_{\p}(g_k(\hat{u}_i))>\dep_{\p}(v_i)$ (i.e., $\dep_{\p}(\p'^{(2^k)}(\hat{u}_i))>\dep_{\p}(v_i)$), we set $\hat{u}_i\gets g_k(\hat{u}_i)=\p'^{(2^k)}(\hat{u}_i)$. Due to Lemma~\ref{lem:properties_compressed_rooted_tree} again, the final $\hat{u}_i$ must satisfy $\dep_{\p}(\hat{u}_i)\geq\dep_{\p}(v_i)$ and $\dep_{\p}(\hat{u}_i)\leq\dep_{\p}(v_i)+2t$. This step takes time $O(t)$.
\end{enumerate}
\begin{lemma}[LCA algorithm]\label{lem:correctness_of_LCA}
Let $\p:V\rightarrow V$ be a set of parent pointers on a vertex set $V$. $\p$ has a unique root. Let $Q=\{(u_1,v_1),(u_2,v_2),\cdots,(u_q,v_q)\}$ be a set of $q$ pairs of vertices where $\forall i\in[q],u_i\not =v_i,u_i,v_i\in \leaves(\p)$. Let $\lca:Q\rightarrow V\times V\times V$ be the output of \textsc{LCA}$(\p,Q)$. For $(u_i,v_i)\in Q$, $(p_i,p_{i,u_i},p_{i,v_i})=\lca(u_i,v_i)$ satisfies that $p_i$ is the LCA of $(u_i,v_i)$, $p_{i,u_i},p_{i,v_i}$ are ancestors of $u_i,v_i$ respectively, and $p_{i,u_i},p_{i,v_i}$ are children of $p_i$.
\end{lemma}
\begin{proof}
Without loss of generality, we can assume $\dep_{\p}(u_i)\geq \dep_{\p}(v_i)$.
After step~\ref{sta:find_a_close_u}, $\hat{u}_i$ satisfies $\dep_{\p}(\hat{u}_i)\geq \dep_{\p}(v_i)$ and $\dep_{\p}(\hat{u}_i)\leq \dep_{\p}(v_i)+2t$. Notice that the LCA of $(u_i,v_i)$ in $\p$ is the same as the LCA of $(\hat{u}_i,v_i)$ in $\p$. In step~\ref{sta:check_directly}, if we find the LCA of $(\hat{u}_i,v_i)$, then the lemma holds for $\lca(u_i,v_i)$. Otherwise, the depth of the LCA of $(\hat{u}_i,v_i)$ is smaller than $\dep_{\p}(\hat{u}_i)-4t\leq \dep_{\p}(v_i)-2t$. By combining with Lemma~\ref{lem:properties_compressed_rooted_tree}, neither of $u'_i$ nor $v'_i$ in step~\ref{sta:find_vertex_in_Vprime} can be the LCA of $(\hat{u}_i,v_i)$ in $\p$. Thus, the LCA of $(u_i,v_i)$ in $\p$ is the same as the LCA of $(u'_i,v'_i)$ in $\p$.
According to step~\ref{sta:find_uv_double_prime}, $u''_i,v''_i$ are ancestors of $u'_i,v'_i$ respectively in both $\p$ and $\p'$, but neither of $u''_i$ nor $v''_i$ is the common ancestor of $(u'_i,v'_i)$. Furthermore, $\p'(u''_i)=\p'(v''_i)$ is the LCA of $u'_i,v'_i$ in $\p'$. Thus, $\p'(u''_i)$ is a common ancestor of $(u'_i,v'_i)$ in $\p$. 
By combining with Lemma~\ref{lem:properties_compressed_rooted_tree}, we know that there exists $j\in[2t]$ such that $\p^{(j)}(u''_i)$ is the LCA of $(u'_i,v'_i)$ in $\p$. In step~\ref{sta:find_final_LCA}, we can find the LCA of $(u'_i,v'_i)$ in $\p$ and thus the LCA of $(u_i,v_i)$.
\end{proof}

\subsection{Multi-Paths Generation}\label{sec:path_generation}
Consider a rooted tree represented by a set of parent pointers $\p:V\rightarrow V$ on a vertex set $V$ and a set of $q$ vertex-ancestor pairs $Q=\{(u_1,v_1),(u_2,v_2),\cdots,(u_q,v_q)\}$ where $\forall i\in[q],$ $v_i$ is an ancestor of $u_i$. We show a space efficient algorithm \textsc{MultiPaths}$(\p,Q)$ which can generate all the paths $P(u_1,v_1),P(u_2,v_2),\cdots,P(u_q,v_q)$. 

\begin{figure}
\noindent\fbox{
	\begin{minipage}{\linewidth}
		Multi-Paths Generation:
		\begin{itemize}
			\item \textbf{Input}:
			\begin{itemize}
				
				\item A rooted tree represented by a set of parent pointers $\p:V\rightarrow V$ on a set $V$ of $n$ vertices ($\p$ has a unique root $r$), and a set of $q$ vertex-ancestor pairs $Q=\{(u_1,v_1),(u_2,v_2),\cdots,(u_q,v_q)\}$ where $\forall i\in[q],v_i$ is an ancestor of $u_i$.
				
			\end{itemize}
			
			\item \textbf{Output}:
			\begin{itemize}
				
				\item $P_1,P_2,\cdots,P_q$.
				
			\end{itemize}
			
			\item \textbf{Generating multiple path sequences} (\textsc{MultiPaths}$(\p:V\rightarrow V,Q)$ ): 
			\begin{enumerate}
				\item $(V',\p')\gets$\textsc{Compress}$(\p)$. \hfill{//(see Lemma~\ref{lem:properties_compressed_rooted_tree}).}
				\item Set $d\gets \dep(\p), t\gets \lceil\log d\rceil$ and compute mappings $g_0,g_1,\cdots g_{t}:V'\rightarrow V'$ such that $\forall v\in V',j\in\{0,1,\cdots,t \}$, $g_j(v)=\p'^{(2^j)}(v)$.
				\item For each vertex-ancestor pair $(u_i,v_i)\in Q$:\label{sta:multiple_path_queries}
				\begin{enumerate}
					\item If $\dep_{\p}(u_i)-\dep_{\p}(v_i)\leq 2t$, generate the path sequence $P_i=(u_i,\p^{(1)}(u_i),\p^{(2)}(u_i),\cdots,v_i)$ directly. \label{sta:generate_path_directly}
					\item Otherwise, find the minimum $j\in[2t]$ such that $\p^{(j)}(u_i)\in V'$. Set $u'_i\gets \p^{(j)}(u_i)$. Find an ancestor $v'_i$ of $u'_i$ in $\p'$ such that $\dep_{\p}(v'_i)\geq \dep_{\p}(v_i)$ and $\dep_{\p}(v'_i)-2t\leq \dep_{\p}(v_i)$.
					\label{sta:find_end_points_subpath_in_Vprime}
					\item Generate the path $P'(u'_i,v'_i)$ in $\p'$.\label{sta:find_path_pivots}
					\item Initialize a sequence $A$ as the concatenation of $(u_i)$, $P'(u'_i,v'_i)$ and $(v_i)$.\label{sta:initialization_of_path}
					\item Repeat: for each element $a_i$ in $A$, if $a_i$ is not the last element and $a_{i+1}\not=\p(a_i)$, insert $\p(a_i)$ between $a_i$ and $a_{i+1}$; until $A$ does not change. Output the final sequence $A$ as the path sequence $P_i$.\label{sta:final_path}
				\end{enumerate}
			\end{enumerate}
		\end{itemize}
	\end{minipage}
}
\end{figure}

Before we analyze the correctness of the algorithm, let us discuss some details.
\begin{enumerate}
\item In step~\ref{sta:generate_path_directly}, if the length of the path is at most $2t$, then we can generate the path in $O(t)$ rounds. In the $j$-th round, we can find the vertex $\p^{(j)}(u_i)=\p(\p^{(j-1)}(u_i))$.
\item In step~\ref{sta:find_end_points_subpath_in_Vprime}, we use the following way to find $v'_i$. We initialize $v'_i$ as $u'_i$. For $k=t\rightarrow 0$, if $\dep_{\p}(g_k(v'_i))>\dep_{\p}(v_i)$ (i.e., $\dep_{\p}(\p'^{(2^k)}(v'_i))>\dep_{\p}(v_i)$), we set $v'_i\gets g_k(v'_i)=\p'^{(2^k)}(v'_i)$.
\end{enumerate}

\begin{lemma}[Generation of multiple paths]\label{lem:correctness_of_multipath}
Let $\p:V\rightarrow V$ be a set of parent pointers on a vertex set $V$. $\p$ has a unique root. Let $Q=\{(u_1,v_1),(u_2,v_2),\cdots,(u_q,v_q)\}\subseteq V\times V$ be a set of pairs of vertices where $\forall j\in[q],$ $v_j$ is an ancestor of $u_j$ in $\p$. Let $P_1,P_2,\cdots,P_q$ be the output of \textsc{MultiPaths}$(\p,Q)$. Then $\forall j\in [q], P_j=P(u_j,v_j)$, i.e., $P_j$ is a sequence which denotes a path from $u_j$ to $v_j$ in $\p$.
\end{lemma}
\begin{proof}
Consider a pair $(u_i,v_i)\in Q$. If $\dep_{\p}(u_i)-\dep_{\p}(v_i)\leq 2t$, then $P_i$ will be the path from $u_i$ to $v_i$ in $\p$ by step~\ref{sta:generate_path_directly}.

We only need to consider the case when $\dep_{\p}(u_i)>\dep_{\p}(v_i)+2t$. 
According to Lemma~\ref{lem:properties_compressed_rooted_tree}, $\exists j\in[2t]$ such that $\p^{(j)}(u_i)\in V'$. 
Thus, $u'_i\in V'$ can be found by step~\ref{sta:find_end_points_subpath_in_Vprime}. 
Then $v'_i$ can be found.
$v_i$ is an ancestor of $v'_i$.
$v'_i$ is an ancestor of $u'_i$.
$u'_i$ is an ancestor of $u_i$.
In step~\ref{sta:initialization_of_path}, the initialization of $A$ should be 
$
(u_i,u'_i,\p'^{(1)}(u'_i),\p'^{(2)}(u'_i),\cdots,v'_i,v_i).
$
By Lemma~\ref{lem:properties_compressed_rooted_tree}, the initialization of $A$ is also $(u_i,u'_i,\p^{(t)}(u'_i),\p^{(2t)}(u'_i),\cdots,v'_i,v_i)$.
Then by step~\ref{sta:final_path}, the final sequence $P_i=A$ will be $(u_i,\p^{(1)}(u_i),\p^{(2)}(u_i),\cdots,v_i)$ which denotes the path from $u_i$ to $v_i$ in $\p$.
\end{proof}

\subsection{Implementation of the DFS Sequence Algorithm in MPC}\label{sec:implement_DFS}
Here, we discuss how to implement the subroutines mentioned in Section~\ref{sec:compressed_rooted_tree},~\ref{sec:LCA},~\ref{sec:path_generation} in the MPC model.
See section~\ref{sec:data_organize} for the organization of the data in the MPC model and basic MPC operations.

\noindent\textbf{Compressed rooted tree.} Consider the implementation of \textsc{Compress}$(\p:V\rightarrow V)$ (Section~\ref{sec:compressed_rooted_tree}) in the MPC model. The input size is $|V|=n$. In the first step, we need to compute the depth of every vertex in $\p$. As shown by~\cite{asswz18}, this can be computed in the MPC model with $O(n)$ total space and $\Theta(n^{\delta})$ local memory size per machine for any constant $\delta\in(0,1)$ in $O(\log(\dep(\p)))$ time. In the next step, $V'$ can be computed in $O(1)$ time. Finally, we can simultaneously compute $\p'(v)$ for every vertex $v\in V'$. Since $\p'(v)=\p^{(t)}(v)$ for $t=\lceil\log(\dep(\p))\rceil$, it takes $O(t)=O(\log(\dep(\p)))$ time. Therefore, \textsc{Compress}$(\p)$ can be implemented in the $(0,\delta)$-MPC model for any constant $\delta\in(0,1)$ in $O(\log(\dep(\p)))$ time.

\noindent\textbf{Least common ancestor.} Consider the implementation of \textsc{LCA}$(\p:V\rightarrow V,Q)$ (Section~\ref{sec:LCA}) in the MPC model. 
The input size is $|V|+|Q|=n+q$. 
The first step computes a compressed rooted tree $\p':V'\rightarrow V'$. As discussed in the previous paragraph, this only requires $O(n)$ total space and $\Theta(n^{\delta})$ local memory per machine for any constant $\delta\in(0,1)$. Before the next step, we need to compute the depth of each vertex in $\p$ and the depth of each vertex in $\p'$. Since $\dep(\p')\leq \dep(\p)$, it takes $O(\log(\dep(\p)))$ time. 
In step~\ref{sta:find_sparse_table}, as shown in~\cite{asswz18}, $g_0(\cdot)\equiv\p'^{(2^0)}(\cdot),g_1\equiv\p'^{(2^1)}(\cdot),\cdots,g_t\equiv\p'^{(2^t)}(\cdot):V'\rightarrow V'$ for $t=\lceil\log(\dep(\p))\rceil$ can be computed in the MPC model with $O(|V'|\log(\dep(\p')))$ total space and $O(|V'|^{\delta})$ local memory per machine for any constant $\delta\in(0,1)$ in $O(\log(\dep(\p')))=O(\log(\dep(\p)))$ time.
According to Lemma~\ref{lem:properties_compressed_rooted_tree}, $|V'|\leq |V|/\log(\dep(\p))$. 
Thus, step~\ref{sta:find_sparse_table} only needs $O(n)$ total space and takes time $O(\log(\dep(\p)))$.
For step~\ref{sta:handle_each_lca_query}, we can handle all the queries in $Q$ simultaneously.
For step~\ref{sta:find_a_close_u}, we can use $O(1)$ time to check whether $\dep_{\p}(u_i)>\dep_{\p}(v_i)+2t$. If it is true, we can use $O(t)=O(\log(\dep(\p)))$ time to find a $j\in\{0,1,\cdots,2t\}$ such that $\p^{(j)}(u_i)\in V'$.
Then, we apply an exponential search by using $g_0,g_1,\cdots,g_t$ to find $\hat{u}_i$. 
This takes $O(t)=O(\log(\dep(\p)))$ time.
Step~\ref{sta:check_directly} checks whether $\p^{(j)}(\hat{u}_i)$ is the LCA for every $j\in[4t]$. 
Thus, it takes $O(t)=O(\log(\dep(\p)))$ time.
In step~\ref{sta:find_vertex_in_Vprime}, according to Lemma~\ref{lem:properties_compressed_rooted_tree}, there exists $j\in\{0,1,2,\cdots,2t\}$ such that $\p^{(j)}(\hat{u}_i)\in V'$. Thus, we only need time $O(t)$ to find $u'_i$. Similarly, we only need time $O(t)$ to find $v'_i$.
In step~\ref{sta:find_uv_double_prime},
by~\cite{asswz18}, the LCA of each $(u'_i,v'_i)$ in $\p'$ can be computed simultaneously in the MPC model with $O(|V'|\log |V'| + |Q|)=O(n)$ total space in $O(\log(\dep(\p')))=O(\log(\dep(\p)))$ time.
The last step checks whether $\p^{(j)}(u''_i)=\p^{(j)}(v''_i)$ for each $j\in[2t]$. Thus it requires $O(t)=O(\log(\dep(\p)))$ time.
To conclude, \textsc{LCA}$(\p:V\rightarrow V,Q)$ can be implemented in the $(0,\delta)$-MPC model for any constant $\delta\in(0,1)$ in $O(\log(\dep(\p)))$ parallel time.

\noindent\textbf{Multiple paths generation.} Consider the implementation of \textsc{MultiPaths}$(\p:V\rightarrow V,Q)$ (Section~\ref{sec:path_generation}) in the MPC model. The first two steps are the same as the first two in the LCA subroutine mentioned in the previous paragraph. 
They can be implemented in the MPC model with $O(|V|)=O(n)$ total space and $\Theta(n^{\delta})$ local memory per machine for any constant $\delta\in(0,1)$ in $O(\log(\dep(\p)))$ time.
We compute the depth of each vertex in $\p$ and the depth of each vertex in $\p'$ in $O(\log(\dep(\p)))$ time before the next step. 
In step~\ref{sta:multiple_path_queries}, all the queries $(u_i,v_i)\in Q$ can be handled simultaneously.
In step~\ref{sta:generate_path_directly}, if $\dep_{\p}(u_i)\leq \dep_{\p}(v_i)+ 2t$, the length of the path from $u_i$ to $v_i$ is at most $2t$, and thus $P(u_i,v_i)$ can be computed in $O(t)=O(\log(\dep(\p)))$ time.
In step~\ref{sta:find_end_points_subpath_in_Vprime}, we can use $O(t)=O(\log(\dep(\p)))$ time to find the minimum $j\in[2t]$ such that $\p^{(j)}(u_i)\in V'$. 
Then we can apply exponential search to find $v'_i$ by using $g_0,g_1,\cdots,g_t$ in $O(t)=O(\log(\dep(\p)))$ time.
In step~\ref{sta:find_path_pivots}, by~\cite{asswz18}, each path $P'(u'_i,v'_i)$ in $\p'$ can be generated simultaneously in the MPC model with $O(|V'|\log |V'| + \sum_{i\in [q]}|P'(u'_i,v'_i)|)=O(n+\sum_{i\in [q]}|P(u_i,v_i)|)$ total space in $O(\log(\dep(\p')))=O(\log(\dep(\p)))$ time.
Consider the initialization of $A=(a_1,a_2,\cdots,a_h)$ in step~\ref{sta:initialization_of_path}. 
$a_1$ should be $u_i$ and $a_h$ should be $v_i$.
By Lemma~\ref{lem:properties_compressed_rooted_tree}, $\forall j\in[h-1]$, $\dep(a_j)-\dep(a_{j+1})\leq 2t$.
Thus, the number of repetitions in the final step is at most $O(t)=O(\log(\dep(\p)))$.
To conclude, \textsc{MultiPaths}$(\p:V\rightarrow V,Q=\{(u_1,v_1),(u_2,v_2),\cdots,(u_q,v_q)\})$ can be implemented in the MPC model with total space linear in $O(|V|+\sum_{i\in [q]} |P(u_i,v_i)|)$ and local memory size $\Theta(|V|^{\delta})$ per machine for any constant $\delta\in(0,1)$ in $O(\log(\dep(\p)))$  time.

\noindent\textbf{DFS sequence in the MPC model.} Consider \textsc{LeafSampling}$(n^{\delta},\p:V\rightarrow V)$ where $n=|V|$ and $\delta$ is an arbitrary constant from $(0,1)$. For step~\ref{sta:finding_LCA} of \textsc{LeafSampling$(n^{\delta},\p)$}, we run our LCA (Section~\ref{sec:LCA}) algorithm.
The correctness of our LCA algorithm is guaranteed by Lemma~\ref{lem:correctness_of_LCA}.
 According to~\cite{asswz18}, the total number of queries generated in step~\ref{sta:finding_LCA} of \textsc{LeafSampling$(n^{\delta},\p)$} is at most $O(n^{\delta})$ with high probability. Then due to the discussion in the previous paragraphs, the step~\ref{sta:finding_LCA} of \textsc{LeafSampling$(n^{\delta},\p)$} can be implemented in the $(0,\delta)$-MPC model for any constant $\delta\in(0,1)$ in $O(\log(\dep(\p)))$ time.
 For step~\ref{sta:finding_paths} of \textsc{LeafSampling$(n^{\delta},\p)$}, we run our multiple paths generation (Section~\ref{sec:path_generation}) algorithm.
 The correctness of our multiple paths generation algorithm is guaranteed by Lemma~\ref{lem:correctness_of_multipath}. 
 Notice that the total length of all the queried paths in the step~\ref{sta:finding_paths} of \textsc{LeafSampling$(n^{\delta},\p)$} is at most the length of the DFS sequence which is $O(n)$.
 According to the discussion in the previous paragraphs, the step~\ref{sta:finding_paths} of  \textsc{LeafSampling$(n^{\delta},\p)$} can be implemented in the $(0,\delta)$-MPC model for any constant $\delta\in(0,1)$ in $O(\log(\dep(\p)))$ time. Together with Theorem~\ref{thm:previous_leaf_sampling}, we conclude Theorem~\ref{thm:DFS_sequence}.

\section{$2$-Edge Connectivity and Biconnectivity in MPC}\label{sec:mpc_biconnectivity}
In this section, we will discuss how to implement the $2$-edge connectivity algorithm and the biconnectivity algorithm in the MPC model.
Let us firstly introduce how to implement an subroutine called range minimum query (RMQ) in the MPC model. 

\subsection{Parallel Range Minimum Query in Linear Total Space}\label{sec:rmq}
The range minimum query (RMQ) problem is as the following. 
Given a sequence $A=(a_1,a_2,\cdots,a_n)$ and a set of queries $Q=\{(l_1,r_1),(l_2,r_2),\cdots,(l_q,r_q)\}$ where $\forall i\in [q], l_i\leq r_i\in [n]$, we want to find the value $\min_{l_i\leq j\leq r_i} a_j$ for each query $(l_i,r_i)\in Q$. 
\cite{asswz18} shows an MPC algorithm which requires total space $O(n\log n+q)$ and takes $O(1)$ parallel time for solving the RMQ problem. 
Their space is not linear in the input size. 
In this section, we show that if every query $(l_i,r_i)\in Q$ satisfies $r_i-l_i\geq 2\lceil\log n\rceil$, then we can solve the such RMQ problem in the MPC model with total space $O(n+q)$ in $O(1)$ parallel time.
The offline description is shown in the algorithm \textsc{RMQ}$(A,Q)$.
\begin{figure}
	\noindent\fbox{
		\begin{minipage}{\linewidth}
			Multiple RMQ Algorithm:
			\begin{itemize}
				\item \textbf{Input}:
				\begin{itemize}
					
					\item An sequence $A=(a_1,a_2,\cdots,a_n)\in \mathbb{Z}^n$ and a set $Q=\{(l_1,r_1),(l_2,r_2),\cdots,(l_q,r_q)\}$, where $\forall i\in[q],l_i,r_i\in[n], l_i+\lceil\log n\rceil\leq r_i$.
					
				\end{itemize}
				
				\item \textbf{Output}:
				\begin{itemize}
					
					\item $\rmq:Q\rightarrow \mathbb{Z}$.
					
				\end{itemize}
				
				\item \textbf{Finding the minimum value in queried ranges} (\textsc{RMQ}$(A,Q)$ ): \label{sta:rmq_qureis_in_parallel}
				\begin{enumerate}
					\item Set $t\gets \lceil\log n\rceil$. Set $A'\gets (a_1',a_2',\cdots,a_{\lceil n/t\rceil}')$, where 
					$$
					\forall i\in[\lceil n/t\rceil],a'_i\gets \min_{j\in[n]:(i-1)\cdot t< j\leq i\cdot t} a_j.
					$$\label{sta:compressed_sequence}
					\item Initialize $\lef:[n]\rightarrow \mathbb{Z},\rig:[n]\rightarrow \mathbb{Z}$. 
					For each $i\in [n]$, find $j\in [\lceil n/t\rceil]$ such that $i\in((j-1)t,jt]$. Set $\lef(i) \gets \min_{k\in[n]\cap((j-1)t,i]} a_k,\rig(i)\gets\min_{k\in[n]\cap[i,jt]} a_k$.
					\item For each $(l_i,r_i)\in Q$:\label{sta:rmq_qureis_in_parallel}
					\begin{enumerate}
						\item Find the smallest $l_i'\geq l_i$ with $l_i'\mod t=0$ and find the largest $r_i'\leq r_i$ with $r_i'\mod t =0$.\label{sta:inner_interval}
						\item If $l_i'=r_i'$, set $m_i\gets \infty$; otherwise $m_i\gets \min_{l_i'/t + 1\leq j\leq r_i'/t} a'_j$. \label{sta:find_in_middle}.
						\item Set $\rmq((l_i,r_i))\gets \min(\rig(l_i),m_i,\lef(r_i))$.\label{sta:rmq_output}
					\end{enumerate}
				\end{enumerate}
			\end{itemize}
		\end{minipage}
	}
\end{figure}
\begin{lemma}[Range minimum query]\label{lem:rmq}
Let $A=(a_1,a_2,\cdots,a_n)\in \mathbb{Z}^n$ be a sequence of $n$ numbers and $Q=\{(l_1,r_1),(l_2,r_2),\cdots,(l_q,r_q)\}$ where $\forall i\in[q],l_i,r_i\in [n],l_i+\lceil\log n\rceil \leq r_i$. 
Let $\rmq:Q\rightarrow \mathbb{Z}$ be the output of \textsc{RMQ}$(A,Q)$. Then $\forall (l_i,r_i)\in Q$, $\rmq((l_i,r_i)) = \min_{j\in [n]\cap [l_i,r_i]} a_j.$ In addition, $\textsc{RMQ}$ can be implemented in the $(0,\delta)$-MPC model for any constant $\delta\in(0,1)$ in $O(1)$ parallel time. 
\end{lemma}
\begin{proof}
Firstly, let us consider the correctness of \textsc{RMQ}$(A,Q)$. 
Let $t=\lceil\log n\rceil$.
For a query $(l_i,r_i)\in Q$, since $l_i+t\leq r_i$, the $l'_i,r'_i$ found by the step~\ref{sta:inner_interval} will satisfy $l'_i\leq r'_i$.
If $l'_i=r'_i$, then $m_i=\infty$ and $\rmq((l_i,r_i))=\min(\min_{l_i\leq j\leq l_i'} a_j,\min_{l_i'\leq j\leq r_i}a_j) = \min_{l_i\leq j\leq r_i} a_j $.
Otherwise, by step~\ref{sta:find_in_middle}, $m_i=\min_{l_i'+1\leq j\leq r_i'} a_j$.
By step~\ref{sta:rmq_output}, $\rmq((l_i,r_i))=\min(\min_{l_i\leq j\leq l_i'} a_j,\min_{l_i'+1\leq j\leq r_i'} a_j,\min_{r_i'\leq j\leq r_i}a_j)= \min_{l_i\leq j\leq r_i} a_j$.

Let us analyze the total space required and the parallel time for running \textsc{RMQ}$(A,Q)$ in the MPC model.
According to Theorem~\ref{thm:MPC_sorting}, the sorting takes $O(1)$ time and requires linear total space. 
Notice that $\delta\in(0,1)$ is a constant and each machine has $\Theta(n^{\delta})$ local memory. 
We can sort $a_1,a_2,\cdots,a_n$ by their indexes and $o(n)$ number of duplicates of some elements in $A$ such that $a_{i\cdot n^{\delta}+1},\cdots,a_{(i+1)\cdot n^\delta},a_{(i+1)\cdot n^\delta+1},\cdots,a_{(i+1)\cdot n^\delta+t}$ are on the $i^{\text{th}}$ machine. 
Therefore, the first two steps of \textsc{RMQ}$(A,Q)$ can be implemented in the MPC model with $O(n)$ total space and in time $O(1)$.
For step~\ref{sta:rmq_qureis_in_parallel}, we can handle all the queries $(l_i,r_i)\in Q$ simultaneously.
Step~\ref{sta:inner_interval} only requires local computations.
Step~\ref{sta:find_in_middle} needs to handle at most $|Q|$ RMQ on the sequence $A'$. 
Due to~\cite{asswz18}, this can be implemented in the MPC model with $O(|A'|\log |A'|+|Q|)=O(n+q)$ total space and $O(1)$ parallel time.
Step~\ref{sta:rmq_output} can be done in $O(1)$ time.
To conclude, \textsc{RMQ}$(A,Q)$ can be implemented in the $(0,\delta)$-MPC model for any constant $\delta\in(0,1)$ and the parallel time is $O(1)$.
\end{proof}

\subsection{MPC Implementation of $2$-Edge Connectivity and Biconnectivity}
The input is a connected undirected graph $G=(V,E)$. $G$ has $|V|=n$ vertices and $|E|=m$ edges. 
Thus, the input size is $m+n$.
Consider the $(\gamma,\delta)$-MPC model for $\gamma\in[0,2]$ and an arbitrary constant $\delta\in(0,1)$. 
The total space in the system should be $\Theta(m^{1+\gamma})$ and the local memory size of each machine is $\Theta(m^{\delta})$.
There is an efficient algorithm for solving connected components and spanning tree problem.

\begin{theorem}[\cite{asswz18}]\label{thm:spanning_tree}
For any $\gamma\in[0,2]$ and any constant $\delta\in(0,1)$, there is a randomized $(\gamma,\delta)$-MPC algorithm which outputs the connected components together with a rooted spanning forest of an undirected graph $G$ with $n$ vertices and $m$ edges in $O(\min(\log \diam(G)\cdot \log\frac{\log n}{\log ((n+m)^{1+\gamma}/n)},\log n))$ parallel time. Furthermore, the depth of the spanning forest is at most $\min\left(\diam(G)^{O\left(\log\frac{\log n}{\log ((n+m)^{1+\gamma}/n)}\right)},n\right)$. The success probability is at least $0.98$. If the algorithm fails, then it returns FAIL.
\end{theorem}

\noindent\textbf{$2$-Edge connectivity.} In the first step of \textsc{Bridges}$(G)$ (Section~\ref{sec:two_edge_connectivity}), according to Theorem~\ref{thm:spanning_tree}, with probability $0.98$, the rooted spanning tree of $G$ can be computed in the MPC model with total space $O(m^{1+\gamma})$ in $O(\log \diam(G)\cdot \log\log_{m^{1+\gamma}/n }n)$ time, and
the depth of the spanning tree is at most $\diam(G)^{O(\log\log_{m^{1+\gamma}/n }n)}$.
In step~\ref{sta:bridge_set_bac}, to compute $\bac(v)$ for each $v\in V$, we can query the LCA of $(v,w)$ in $\p$ for each edge $(v,w)\in E$. 
We can use our LCA algorithm (Section~\ref{sec:LCA}) as the subroutine for this purpose. It takes the total space $O(m)$ and the running time $O(\log(\dep(\p)))=O(\log \diam(G)\cdot \log\log_{m^{1+\gamma}/n }n)$ (Section~\ref{sec:implement_DFS}).
In step~\ref{sta:two_edge_dfs_sequence}, with probability at least $0.99$, the DFS sequence can be computed using $O(n)$ total space in time $O(\log(\dep(\p)))=O(\log \diam(G)\cdot \log\log_{m^{1+\gamma}/n }n)$ (Theorem~\ref{thm:DFS_sequence}).
In step~\ref{sta:final_output_bridge}, we can use sorting to find the first appearance $a_i$ and the last appearance $a_j$ in the DFS sequence of each vertex $v$, and $\min_{k\in\{i,i+1,\cdots,j\}}\bac(a_k)$ corresponds to a range minimum query.
If the size of the subtree of $v$ is at most $\log n$, the corresponding RMQ can be solved by local computation.
Otherwise, we use our RMQ algorithm (Section~\ref{sec:rmq}) to handle the corresponding RMQ of $v$.
By Lemma~\ref{lem:rmq}, this step only takes $O(1)$ time and requires $O(n)$ space. To conclude, \textsc{Bridges}$(G)$ only takes total space $O(m^{1+\gamma})$ and has parallel time $O(\log \diam(G)\cdot \log\log_{m^{1+\gamma}/n }n)$. 

Since the correctness of \textsc{Bridges}$(G)$ (Section~\ref{sec:two_edge_connectivity}) is guaranteed by Lemma~\ref{lem:two_edge_connect}, we can conclude Theorem~\ref{thm:two_edge_connectivity}.

\noindent\textbf{Biconnectivity.} The first three steps of \textsc{Biconn}$(G)$ (Section~\ref{sec:biconn}) are the same as the first three steps of \textsc{Bridges}$(G)$ (Section~\ref{sec:two_edge_connectivity}). 
Thus, the success probability of the first three steps is at least $0.97$. 
The total space used is at most $O(m^{1+\gamma})$ and the running time is at most $O(\log \diam(G)\cdot \log\log_{m^{1+\gamma}/n }n)$.
Step~\ref{sta:non_tree_edge_ancestor} of \textsc{Biconn}$(G)$ corresponds to the RMQ problem which is almost the same as the step~\ref{sta:final_output_bridge} of \textsc{Bridges}$(G)$.
Thus, it takes $O(n)$ total space and $O(1)$ parallel time.
Step~\ref{sta:non_tree_edge_non_ancestor} requires $m$ LCA queries. 
We can run our LCA algorithm (Section~\ref{sec:LCA}) for this step.
It takes $O(m+n)$ space and $O(\log(\dep(\p)))=O(\log \diam(G)\cdot \log\log_{m^{1+\gamma}/n }n)$ time (Section~\ref{sec:implement_DFS}).
By Lemma~\ref{lem:biconnectivity}, we have $\diam(G')\leq \diam(G)^{O(\log\log_{m^{1+\gamma}/n} n)}\cdot \bidiam(G)$.
According to Theorem~\ref{thm:spanning_tree}, with probability at least $0.98$, the connected components of $G'$ can be computed in step~\ref{sta:connected_components_Gprime}, the total space needed is $O(m^{1+\gamma})$, and the running time is $O(\log\diam(G)\log^2\log_{m^{1+\gamma}/n}n+\log\bidiam(G)\log\log_{m^{1+\gamma}/n}n)$.
To conclude, the total space needed is at most $O(m^{1+\gamma})$, and the parallel running time is $O(\log\diam(G)\log^2\log_{m^{1+\gamma}/n}n+\log\bidiam(G)\log\log_{m^{1+\gamma}/n}n)$.

Since the correctness of \textsc{Biconn}$(G)$ (Section~\ref{sec:biconn}) is guaranteed by Lemma~\ref{lem:biconnectivity}, we can conclude Theorem~\ref{thm:final_biconn}.

\section{Hardness of Biconnectivity in MPC}\label{sec:hardness_of_biconnectivity}
There is a conjectured hardness result which is widely used in the MPC literature~\cite{ksv10,bks13,klmrv14,rvw16,yv18}. 

\begin{conjecture}[One cycle vs. two cycles]\label{con:onevstwo}
For any $\gamma\geq 0$ and any constant $\delta\in(0,1)$, distinguishing the following two graph instances in the $(\gamma,\delta)$-MPC model requires $\Omega(\log n)$ parallel time: 
\begin{enumerate}
	\item a single cycle contains $n$ vertices,
	\item two disjoint cycles, each contains $n/2$ vertices.
\end{enumerate}
\end{conjecture}

Under the above conjecture, we show that $\Omega(\log \bidiam(G))$ parallel time is necessary to compute the biconnected components of $G$. This claim is true even for the constant diameter graph $G$, i.e., $\diam(G)=O(1)$.

\begin{theorem}[Hardness of biconnectivity in MPC]
For any $\gamma\geq 0$ and any constant $\delta\in(0,1)$, unless the one cycle vs. two cycles conjecture (Conjecture~\ref{con:onevstwo}) is false, any $(\gamma,\delta)$-MPC algorithm requires $\Omega( \log \bidiam(G))$ parallel time for testing whether a graph $G$ with a constant diameter is biconnected.  
\end{theorem}
\begin{proof}
For $\gamma\geq 0$ and an arbitrary constant $\delta\in(0,1)$, suppose there is a $(\gamma,\delta)$-MPC algorithm $\mathcal{A}$ which can determine whether an arbitrary constant diameter graph $G$ is biconnected in $o(\log \bidiam(G))$ parallel time. Then we give a $(\gamma,\delta)$-MPC algorithm for solving one cycle vs. two cycles problem as the following:
\begin{enumerate}
	\item For a one cycle vs. two cycles instance $n$-vertex graph $G'=(V',E')$, construct a new graph $G=(V,E)$: $V=V'\cup\{v^*\},E=E'\cup\{(v,v^*)\mid v\in V'\}$.
	\item Run $\mathcal{A}$ on $G$. If $G$ is not biconnected, $G'$ contains two cycles. Otherwise $G'$ is a single cycle.
\end{enumerate}
It is easy to see that the diameter of $G$ is $2$.
If $G'$ is a single cycle, then $G$ is biconnected and $\bidiam(G)=\Theta(n)$. If $G'$ contains two cycles, then $G$ contains two biconnected components and $\bidiam(G)=\Theta(n)$.

The first step of the above algorithm takes $O(1)$ parallel time and only requires linear total space. The graph $G$ has $n+1$ vertices and $2n$ edges.
Thus, the above algorithm is also a $(\gamma,\delta)$-MPC algorithm.  The parallel time of the above algorithm is the same as the time needed for running $\mathcal{A}$ on $G$ which is $o(\log\bidiam(G))=o(\log n)$. 
Thus the existence of the algorithm $\mathcal{A}$ implies that the one cycle vs. two cycles conjecture (Conjecture~\ref{con:onevstwo}) is false. 
\end{proof}
\bibliography{ref}

\end{document}